\documentclass[final, authorcolumns]{lipics-v2021}

\usepackage{amsmath}
\usepackage{amssymb}
\usepackage{bbm}
\usepackage{cleveref}
\usepackage{mathtools}
\usepackage{nicefrac}
\usepackage{stmaryrd}
\usepackage{adjustbox}
\usepackage{tikz}
\usepackage{tikz-cd}
\usetikzlibrary{shapes}

\usepackage[disable]{todonotes} 

\newcommand{\short}[1]{}\newcommand{\full}[1]{#1}

\short{\usepackage[appendix=strip,bibliography=common]{apxproof}}
\full{\usepackage[appendix=append,bibliography=common]{apxproof}}

\newtheoremrep{theorem}{Theorem}
\newtheoremrep{lemma}[theorem]{Lemma}
\newtheoremrep{proposition}[theorem]{Proposition}
\newtheoremrep{corollary}[theorem]{Corollary}
\theoremstyle{definition}
\newtheoremrep{definition}[theorem]{Definition}
\newtheoremrep{example}[theorem]{Example}
\newtheoremrep{remark}[theorem]{Remark}

\newcommand{\Vcat}{\mathcal{V}\textrm{-}\mathsf{Cat}}

\newcommand{\Vgraph}{\mathcal{V}\textrm{-}\mathsf{Graph}}
\newcommand{\Vrel}{\mathcal{V}\textrm{-}\mathsf{Rel}}
\newcommand{\Vspred}{\mathcal{V}\textrm{-}\mathsf{SPred}}
\newcommand{\Set}{\mathsf{Set}}

\newcommand{\Pos}{\mathsf{Pos}}
\newcommand\VV{\mathcal{V}}
\newcommand{\ev}{\mathit{ev}}
\newcommand{\id}{\mathrm{id}}
\newcommand{\Id}{\mathrm{Id}}

\newcommand{\op}{\text{op}}

\newcommand{\pfun}{\mathcal{P}}
\newcommand{\dfun}{\mathcal{D}}
\newcommand{\expect}[1]{\mathbb{E}_{#1}}
\newcommand{\nonexp}[2]{{#1} \to {#2}}

\newcommand{\by}[1]{(\text{#1})}

\newcommand{\constfun}[1]{C_{#1}}

\newcommand{\vle}{\sqsubseteq}
\newcommand{\vge}{\sqsupseteq}
\newcommand{\vJoin}{\bigsqcup}
\newcommand{\vMeet}{\bigsqcap}
\newcommand{\vjoin}{\sqcup}
\newcommand{\vmeet}{\sqcap}

\newcommand{\kant}[1]{K_{#1}}

\newcommand{\twoQ}{{2_\vmeet}}
\newcommand{\unitQ}{{[0,1]_\oplus}}
\newcommand{\zeroinfQ}{{[0,\infty]_+}}

\newcommand{\C}{\mathbf{C}}

\newcommand{\supp}{\textrm{supp}}

\newcommand{\sprod}{\textstyle\prod}

\author{Keri D'Angelo}{Cornell University, USA}{kd349@cornell.edu}{https://orcid.org/0000-0002-8812-9612}{}
\author{Sebastian Gurke}{Universit\"at Duisburg-Essen, Germany}{sebastian.gurke@uni-due.de}{https://orcid.org/0009-0008-4343-1384}{}
\author{Johanna Maria Kirss}{University of Copenhagen, Denmark}{rmg215@alumni.ku.dk}{}{was supported by NSF grant DMS-2216871}
\author{Barbara K\"onig}{Universit\"at Duisburg-Essen, Germany}{barbara_koenig@uni-due.de}{https://orcid.org/0000-0002-4193-2889}{}
\author{Matina Najafi}{Amirkabir University of Technology, Iran}{matn@aut.ac.ir}{}{}
\author{Wojciech R{\'o}{\.z}owski}{University College London, UK}{w.rozowski@cs.ucl.ac.uk}{https://orcid.org/0000-0002-8241-7277}{partially supported by ERC grant Autoprobe (no. 101002697)}
\author{Paul Wild}{FAU Erlangen-N\"urnberg, Germany}{paul.wild@fau.de}{https://orcid.org/0000-0001-9796-9675}{}

\authorrunning{K. D'Angelo, S. Gurke, J. Kirss, B. K\"onig, M. Najafi, W. R{\'o}{\.z}owski and P. Wild}
\Copyright{Keri D'Angelo, Sebastian Gurke, Johanna Maria Kirss, Barbara K\"onig, Matina Najafi, Wojciech R{\'o}{\.z}owski and Paul Wild}

\title{Behavioural Metrics: Compositionality of the Kantorovich
  Lifting and an Application to Up-To Techniques}
\titlerunning{Compositionality of the Kantorovich Lifting}


\funding{\textit{Sebastian Gurke, Barbara K\"onig and Paul Wild:} were supported by the Deutsche Forschungsgemeinschaft (DFG, German Research Foundation) -- project number 434050016 (SpeQt)}

\acknowledgements{We would like to thank Avi Craimer for helpful
  discussions and the organizers of the ACT Adjoint School for enabling
  us to meet and write this paper together!}

\relatedversion{\full{This paper is the extended version of
    \cite{ws:metrics-fuzzy-lax}.}\short{Proofs and extra material can
    be found in the full version
    \cite{dgkknrw:behavioural-metrics-compositionality-arxiv}.}}

\DOIPrefix{}


\nolinenumbers

\begin{document}

\maketitle


\begin{abstract}
  Behavioural distances of transition systems modelled via coalgebras
  for endofunctors generalize traditional notions of behavioural
  equivalence to a quantitative setting, in which states are equipped
  with a measure of how (dis)similar they are.  Endowing transition
  systems with such distances essentially relies on the ability to
  lift functors describing the one-step behavior of the transition
  systems to the category of pseudometric spaces.  We consider the
  category theoretic generalization of the Kantorovich lifting from
  transportation theory to the case of lifting functors to
  quantale-valued relations, which subsumes equivalences, preorders
  and (directed) metrics.  We use tools from fibred category theory,
  which allow one to see the Kantorovich lifting as arising from an
  appropriate fibred adjunction.  Our main contributions are
  compositionality results for the Kantorovich lifting, where we show
  that that the lifting of a composed functor coincides with the
  composition of the liftings. In addition, we describe how to lift
  distributive laws in the case where one of the two functors is
  polynomial (with finite coproducts). These results are essential
  ingredients for adapting up-to-techniques to the case of
  quantale-valued behavioural distances.  Up-to techniques are a
  well-known coinductive technique for efficiently showing lower
  bounds for behavioural distances. We illustrate the results of our
  paper in two case studies.
\end{abstract}

\section{Introduction}




In concurrency theory, behavioural equivalences are a fundamental
concept: they explain whether two states exhibit the same behaviour
and there are efficient techniques for checking behavioural equivalence
\cite{s:bisimulation-coinduction}. More recently, this idea has been
generalized to behavioural metrics that measure the behavioural
distance of two states
\cite{bw:behavioural-pseudometric,afs:linear-branching-metrics,dgjp:metrics-labelled-markov}. This
is in particular interesting for quantitative systems where the
behaviour of two states might be similar but not identical.

There are two dimensions along which notions of behavioural distances
for (quantitative) transition systems can be generalized: first,
instead of only considering metrics that return a real-valued
distance, one can use an arbitrary quantale, yielding the notion of
quantale-valued relations or conformances that subsume equivalences,
preorders and (directed) metrics. Second, one can abstract from the
branching type of the transition system under consideration, using a
functor specifying various forms of branching (non-deterministic,
probabilistic, weighted, etc.). This leads us to a coalgebraic
framework \cite{r:universal-coalgebra} that provides several
techniques for studying and analyzing systems and has also been
adapted to behavioural metrics
\cite{bbkk:coalgebraic-behavioral-metrics}. A coalgebra is a function
$c\colon X\to FX$ where $X$ is a set of states and $F$ is a (set) functor.

For defining (and computing) behavioural conformances in coalgebraic
generality, a fundamental notion is the lifting of a functor $F$ to
$\overline{F}$, acting on conformances. In the case when $F$ is a
distribution functor and the conformances are metrics, there is a
well-known way of obtaining $\overline{F}$ through results from
transportation theory using either Kantorovich or Wasserstein
liftings, which are known to coincide through the so-called
Kantorovich-Rubinstein duality~\cite{v:optimal-transport}. Recent
work~\cite{bbkk:coalgebraic-behavioral-metrics} generalized these two
approaches to lifting functors to the category of pseudometric spaces,
as well as to more general quantale-valued
relations~\cite{bgkm:hennessy-milner-galois,bkp:up-to-behavioural-metrics-fibrations-journal}. It
turns out that at this level of generality, the analogue of the
Kantorovich-Rubinstein duality does not hold in general anymore. In
both the Kantorovich and Wasserstein approaches, given a set $X$ and a
conformance $d$ (preorder, equivalence, metric, \dots) on $X$, the
lifted functor $\overline{F}$ canonically determines a conformance on
$FX$, based on (a set of) evaluation maps. Intuitively, these maps
provide a way of testing the one-step behaviour of the system and
generalize the calculation of the expected value taking place in the
definitions of the Kantorovich/Wasserstein liftings. Combining the lifting
with a subsequent reindexing with the coalgebra $c$ results in a
function whose (greatest) fixpoint is the desired behavioural
conformance. In this paper we focus on directed conformances such as
preorders or directed metrics (also called hemimetrics).

The aim of this paper is twofold: we consider the so-called
Kantorovich lifting of functors
\cite{bbkk:coalgebraic-behavioral-metrics} that -- as opposed to the
Wasserstein lifting \cite{bkp:up-to-behavioural-metrics-fibrations-journal} -- offers some extra flexibility, since
it allows the use of a set of evaluation maps and places fewer restrictions on both the functor and these predicate liftings. We study
compositionality results for the Kantorovich lifting, answering the
question under which conditions the lifting is compositional in the
sense that $\overline{FG} = \overline{F}\,\overline{G}$. This
compositionality result is then an essential ingredient in adapting
up-to techniques to the setting of behavioural metrics based on the
Kantorovich lifting, inspired by the results of
\cite{bkp:up-to-behavioural-metrics-fibrations-journal}, which were
developed for the Wasserstein lifting. Up-to techniques are
coinductive techniques that allow small witnesses for lower
bounds of behavioural distances by exploiting an algebraic structure
on the state space.

We first set up a framework based on Galois connections and fibred
adjunctions, extending
\cite{bgkmfsw:logics-coalgebra-adjunction}. This sets the stage for
the definition of the Kantorovich lifting based on this adjunction.
We next answer the question of compositionality positively for the
case where $F$ is polynomial with finite coproducts, and show several negative results
for combinations of powerset and distribution functor.

The positive compositionality result for the case where $F$ is
polynomial with finite coproducts opens the door to developing up-to techniques for
trace-like behavioural conformances that are computed on determinized
transition systems, or -- more generally -- determinized
coalgebras. More concretely, we consider coalgebras of type
$c\colon X\to FTX$ where $F$ is a finite coproduct polynomial functor, providing the
explicit branching type of the coalgebra, and $T$ is a monad,
describing the implicit branching. For instance, in the case of a
non-deterministic automaton we would use $FX = 2\times X^A$ and
$T = \mathcal{P}$ (powerset functor). There is a well-known
determinization construction
\cite{jss:trace-determinization-journal,sbbr:generalizing-determinization-journal}
that transforms such a coalgebra into $c^\#\colon TX\to FTX$ via a
distributive law $\zeta\colon TF\Rightarrow FT$. This yields a
determinized system $c^\#$ acting on the state space $TX$, which has
an algebraic structure given by the multiplication of the monad.
A behavioural conformance is then computed on $TX$ and passed back to $X$ using the unit of the monad.

As $TX$ might be very large (e.g., in the case of
$T=\mathcal{P}$) or even infinite (e.g., in the case of
$T=\mathcal{D}$, distribution functor), it can be hard to compute
conformances for $c^\#$.
However, the algebraic structure on $TX$ can be fruitfully employed using
up-to techniques
\cite{p:complete-lattices-up-to,bppr:general-coinduction-up-to} that
allow to consider post-fixpoints up to the algebraic structure, making
it much easier to display a witness proving a (lower) bound on the
distance of two states. The validity of the up-to technique rests on a
compatibility property that is ensured whenever the distributive law
$\zeta$ can be lifted, i.e., if it is non-expansive wrt.\ the lifted
conformances. We show that this holds, where an essential step in the
proof relies on the compositionality results.

In this sense, we complement the up-to techniques in
\cite{bkp:up-to-behavioural-metrics-fibrations-journal} that provide a
similar result for the Wasserstein lifting. This enables us to use the
up-to technique for Kantorovich liftings, which are more versatile and allow more control over the distance values,
in particular in the presence of products and coproducts.
Indeed, as Wasserstein liftings are based on couplings, which in general need not exist~\cite{bbkk:coalgebraic-behavioral-metrics}, they often produce trivial distance values.
We will see later in the paper that several of our case studies
can only be treated appropriately in the Kantorovich
case. Furthermore, we can show the lifting of the distributive law for
an entire class of functors (polynomial functors with finite
coproducts), while
\cite{bkp:up-to-behavioural-metrics-fibrations-journal} contained
generic results for a different class of canonical liftings. In the
non-canonical case it was necessary to prove a complex condition
ensuring compositionality in
\cite{bkp:up-to-behavioural-metrics-fibrations-journal}.

We apply our technique to several examples, such as trace metrics for
probabilistic automata and trace semantics for systems with
exceptions. We give concrete instances where up-to techniques help and
show how witnesses yielding upper bounds can be reduced in size or
even become finite (as opposed to infinite).


\section{Preliminaries}
\label{sec:preliminaries}

We begin by recalling some relevant definitions and fix some notation.


As outlined in the introduction, we use quantales as the domain for
behavioural conformances.  A \emph{quantale} is a complete lattice
$(\VV,\vle)$ that is equipped with a commutative monoid structure
$(\VV,\otimes,k)$ (where $k$ is the unit of $\otimes$) such that
$\otimes$ is \emph{join-continuous}, that is,
$a\otimes\vJoin b_i = \vJoin a\otimes b_i$ for each $a\in\VV$ and each
family $b_i$ in $\VV$, where $\vJoin$ denotes least upper bounds.
Join-continuity of $\otimes$ implies that the operation $a\otimes-$
has a right adjoint, which we denote by $d_\VV(a,-)$.  This means that
we have $a\otimes b\vle c \iff b\vle d_\VV(a,c)$ for all
$a,b,c\in\VV$.  The operation $d_\VV$ is called the \emph{residuation}
or \emph{internal hom} of the quantale.  
\todo{Notation of internal hom? B: I would not do this, because of the
clash with the coproduct notation.}
\begin{toappendix}
  \begin{lemma}
    \label{lem:residuation}
    Residuation (internal hom) in a quantale satisfies the following
    properties for $u,v,w,a,b,c,a_i,b_i\in\VV$:
    \begin{enumerate}
    \item \label{it:residuation-1}
      $d_\VV(v,w) = \bigsqcup \{u\in\VV \mid u\otimes v \vle
      w\}$
    \item \label{it:residuation-2} $d_\VV(v,v) \vge k$ (reflexivity).
    \item \label{it:residuation-3} $d_\VV(u,v) \otimes d_\VV(v,w) \vle d_\VV(u,w)$
      (triangle inequality).
    \item \label{it:residuation-4} $d_\VV(k,w) = w$.
    \item \label{it:residuation-5} $d_\VV(\bot,w) = \top$.
    \item \label{it:residuation-6} $d_\VV(v,\top) = \top$.
    \item \label{it:residuation-7} $d_\VV(\top,\bot) = \bot$.
    \item \label{it:residuation-8} $d_\VV(a,c) \vle d_\VV(d_\VV(b,a),d_\VV(b,c))$
    \item \label{it:residuation-9} $d_\VV(a,c) \vle d_\VV(d_\VV(c,b),d_\VV(a,b))$
    \item \label{it:residuation-10} $\vMeet_{i\in I} d_\VV(a_i,b_i) \vle d_\VV(\vMeet_{i\in I}
      a_i, \vMeet_{i\in I} b_i)$
    \end{enumerate}
  \end{lemma}

  \begin{proof}~
    \begin{enumerate}
    \item Residuation is the right adjoint of the associative
      operation $\otimes$.
    \item
      $d_\VV(v,v) = \bigsqcup \{u\in\VV \mid u\otimes v \vle
      v\} \vge k$, since $k\otimes v = v \vle
      v$.
    \item From $d_\VV(v,w)\vle d_\VV(v,w)$ and the adjunction
      property, we obtain $v\otimes d_\VV(v,w) \vle w$ and
      analogously $u\otimes d_\VV(u,v) \vle v$. Hence
      $u\otimes d_\VV(u,v) \otimes d_\VV(v,w) \vle v\otimes
      d_\VV(v,w) \vle w$. This implies
      $d_\VV(u,v) \otimes d_\VV(v,w) \vle d_\VV(u,w)$.
    \item
      $d_\VV(k,w) = \bigsqcup \{u\in\VV \mid u\otimes k \vle
      w\} = \bigsqcup \{u\in\VV \mid u \vle w\} = w$.
    \item
      $d_\VV(\bot,w) = \bigsqcup \{u\in\VV \mid u\otimes \bot
      \vle w\} = \bigsqcup \{u\in\VV \mid \bot \vle w\}
      = \bigsqcup \VV = \top$. This holds, since $\bot$ is the empty
      join and $\otimes$ preserves joins, hence
      $u\otimes \bot = \bot$.
    \item
      $d_\VV(v,\top) = \bigsqcup \{u\in\VV \mid u\otimes v \vle
      \top\} = \bigsqcup \VV = \top$.
    \item
      $d_\VV(\top,\bot) = \bigsqcup \{u\in\VV \mid u\otimes \top
      \vle \bot\} \vle \bigsqcup \{u\in\VV \mid u\otimes
      k \vle \bot\} = \{u\in\VV \mid u \vle \bot\} =
      \bot$, where the inequality holds since $k\vle \top$.
    \item Recall that because of the definition of residuation, we have the following property:
      \[
        d_{\VV}(a,c) \vle d_{\VV}\left(d_{\VV}(b,a), d_{\VV}(b,c)\right) \iff d_{\VV}(a,c) \otimes d_{\VV}(b,a) \vle d_{\VV}(b,c)
      \]
      Hence, it suffices to show the second inequality:
      \begin{align*}
        d_{\VV}(a,c) \otimes d_{\VV}(b,a) &= d_{\VV}(b,a) \otimes d_{\VV}(a,c) \\
        &\vle d_{\VV}(b,c) \tag{Triangle inequality,
          \Cref{lem:residuation}\eqref{it:residuation-3}}
      \end{align*}
    \item Symmetric to \eqref{it:residuation-8}
    \item Because of the definition of the residuation, for all $i \in I$, we have that:
      \[
        d_{\VV}(a_i, b_i) \vle d_{\VV}(a_i, b_i) \iff d_{\VV}(a_i,b_i) \otimes a_i \vle b_i
      \]
      Therefore, we have that:
      \begin{align*}
        \vMeet_{i \in I} b_i &\vge \vMeet_{i \in I} d_{\VV}(a_i,b_i) \otimes a_i \\
        &\vge \vMeet_{i \in I} d_{\VV}(a_i, b_i) \otimes \vMeet_{i \in I}a_i
      \end{align*}
      Again, by the definition of residuation the above is equivalent
      to
      $\vMeet_{i\in I} d_\VV(a_i,b_i) \vle d_\VV(\vMeet_{i\in I} a_i,
      \vMeet_{i\in I} b_i)$ as desired, which completes the proof. \qedhere
    \end{enumerate}
  \end{proof}
\end{toappendix}

We work with three main examples of quantales.  The first is the
\emph{Boolean quantale} $\twoQ$, consisting of two elements
$\bot\vle\top$ and with binary meet $\vmeet$ as the monoid structure.
In this quantale, the unit $k$ is $\top$, and residuation is Boolean
implication: $d_\VV(a,b) = a \to b = \neg a \vjoin b$.  The second
main example is the \emph{unit interval quantale} $\unitQ$, where the
underlying lattice is the unit interval $[0,1]$ under the reversed
order (that is, ${\vle} = {\ge}$), and with \emph{truncated addition}
$a \oplus b = \min(a+b, 1)$ as the monoid structure.  In this case,
the unit of $\oplus$ is $0$, and its residuation is \emph{truncated
  subtraction}, which is given by
$d_\VV(a,b) = b \ominus a = \max(b-a, 0)$.  The third main example is
the quantale $\zeroinfQ$ of \emph{extended positive reals}, with
structure analogous to $\unitQ$, i.e.\ with reversed lattice order and
using the extended addition (with~$\infty$) as the monoid operation.
\begin{remark}\label{rem:order-reversal}
  As many of our examples use the real-valued quantales $\unitQ$ and $\zeroinfQ$, where the order is reversed, we reserve the use of the symbols $\ge$ and $\le$ for the usual order in the reals, and instead use $\vle$ and $\vge$ when working with general quantales.
  Similarly, we use $\vJoin$ and $\vMeet$ for joins and meets in the quantalic order, but switch to $\inf$ and $\sup$ when working in $\unitQ$ or $\zeroinfQ$.
\end{remark}
We consider several types of conformances based on a quantale $\VV$.
First, given a set $X$, we may consider $\VV$-valued endorelations on
$X$, that is, maps of type $d\colon X\times X\to\VV$.  We call these
structures \emph{$\VV$-graphs} \cite{Lawvere73}, and write $\Vgraph_X$
for the set of $\VV$-graphs with underlying set $X$.  Each set
$\Vgraph_X$ is a complete lattice where both the order and joins are
pointwise, that is $d \vle d'$ if $d(x,y) \vle d'(x,y)$ for all
$x,y\in X$.  Given two $\VV$-graphs $d_X\in\Vgraph_X$ and
$d_Y\in\Vgraph_Y$ we say that a map $f\colon X\to Y$ is
\emph{non-expansive} or a \emph{$\VV$-functor} if
$d_X \vle d_Y\circ (f\times f)$ in $\Vgraph_X$ and in this case we
write $f\colon (X,d_X)\to
(Y,d_Y)$. 
$\VV$-graphs and non-expansive maps form a category $\Vgraph$.
\begin{remark}
  \label{rem:mat-vs-rel}
  In some parts of the literature, the category $\Vgraph$ is denoted
  by $\Vrel$
  instead~\cite{bkp:up-to-behavioural-metrics-fibrations-journal}.  We
  opt for $\Vgraph$ as in \cite{Lawvere73}, as
  $\Vrel$ more often denotes the category with sets as objects and
  $\VV$-valued relations between them as morphisms~\cite{Hofmann07}.
\end{remark}
Second, within $\Vgraph$ we consider the subcategory consisting of
those $\VV$-graphs that satisfy additional axioms, corresponding to a
generalized notion of a metric space, or, equivalently, (small)
categories enriched over $\VV$~\cite{Lawvere73}.  A
\emph{$\VV$-category} is an object of $\Vgraph_X$ for some set $X$
where we additionally have $k \vle d(x,x)$ and
$d(x,y)\otimes d(y,z) \vle d(x,z)$ for all $x,y,z\in X$.  Instantiated
to the quantales $\twoQ$ and $\unitQ$, $\VV$-categories correspond
precisely to \emph{preorders} and \emph{$1$-bounded hemimetric
  spaces}, respectively.  The quantale $\VV$ itself becomes a
$\VV$-category when equipped with the residuation $d_\VV$.  The class
of all $\VV$-categories is denoted by $\Vcat$, and the set of
$\VV$-categories based on a set $X$ is denoted by $\Vcat_X$. 

In this paper, we use a coalgebraic framework. Recall that a {\em
  coalgebra} is a pair $(X, c)$, where $X$ is
the {\em state space} and $c\colon X \to F X$ is the transition structure parametric
on an endofunctor $F \colon \Set \to \Set$. We also
utilize two specific functors for (counter)examples below.
The \emph{powerset functor} is
defined as $\pfun(X) = \{ U \mid U \subseteq X\}$ on sets
and $\pfun (f)(U) = \{f(x) \mid x \in U\}$ on functions.
The \emph{countably supported distribution functor} is defined as
$\dfun(X) = \left\{ p: X \to [0,1] \mid \sum_{x\in X} p(x)=1,
  \textsf{supp}(p) \text{ is countable} \right\}$ on sets, where
$\textsf{supp}(p)= \{ x \mid p(x)\neq 0\}$, and as
$\dfun(f)(p)(y) = \sum \{p(x) \mid x\in X, f(x) = y\}$ on functions.

Given $f_i\colon X\to Y_i$, $i\in\{1,2\}$, we denote by
$\langle g_1,g_2\rangle\colon X\to Y_1\times Y_2$ the mediating
morphism of the product, namely
$\langle g_1,g_2\rangle(x) = (g_1(x),g_2(x))$. Given
$g_i\colon X_i\to Y$, $i\in\{1,2\}$, $[g_1,g_2]\colon X_1+X_2 \to Y$
is the mediating morphism of the coproduct, namely
$[g_1,g_2](x) = g_i(x)$ if $x\in X_i$.


\section{Motivation from Transportation Theory}
\label{sec:transport}
In this section, we give a brief description of the original
Kantorovich distance on probability distributions, before we introduce
its category theoretic generalization. Motivated by the transportation
problem \cite{v:optimal-transport}, the Kantorovich distance on
probability distributions aims to maximize the transport by finding the optimal
flow of commodities that satisfies demand from supplies and minimizes
the flow cost. In fact, it is based on the dual version of this
problem that asks for the optimal price function. For the sake of the
example, consider a metric space defined on a three element set
$X = \{A,B,C\}$ with a distance function
$d \colon X \times X \to [0,\infty]$, such that
$d(A,A) = d(B,B) = d(C,C) = 0$, $d(A,B)=d(B,A)=3$, $d(A,C)=d(C,A)= 5$
and $d(B,C)=d(C,B)=4$\todo{B: drawing?}.  Based on this distance we
now want to define a distance on probability distributions on the set
$X$, which is a function
$d^\uparrow \colon \mathcal{D}(X) \times \mathcal{D}(X) \to
[0,\infty]$.
As a concrete example, let us take the distributions $P$ and $Q$, such
that $P(A)=0.7$, $P(B)=0.1$ and $P(C)=0.2$, while $Q(A)=0.2$,
$Q(B)=0.3$ and $Q(C)=0.5$.

\begin{figure}[!h]
    \centering
    \begin{tikzpicture}[scale=0.9] 
        \coordinate (A) at (0, 0);
        \coordinate (B) at (3, 0); 
        \coordinate (C) at (1.5, 2.0); 

        \node[draw, circle] at (A) (A) {A};
        \node[above=0.1cm of A] {0.7};
        \node[below=0.1cm of A] {0.2};
        
        \node[draw, circle] at (B) (B) {B};
        \node[above=0.1cm of B] {0.1};
        \node[below=0.1cm of B] {0.3};

        \node[draw, circle] at (C) (C) {C};
        \node[above=0.1cm of C] {0.2};
        \node[below=0.1cm of C] {0.5};

        \draw (A) -- (B) node[midway, below] {3};
        \draw (A) -- (C) node[midway, left] {5};
        \draw (B) -- (C) node[midway, right] {4};
    \end{tikzpicture}
    \label{fig:fig1}
\end{figure}

In order to define their distance, we can interpret the three elements
$A$, $B$, $C$ as places where a certain product is produced and
consumed (imagine the places are maple syrup farms, each with an
adjacent café where one can eat the pancakes with the aforementioned
syrup). The geographical distance between the maple syrup farms is
given by the distance function $d$, while the above distributions
model the supply ($Q$) and demand ($P$) of the product in proportion
to the total supply or demand. We assume that the total value of
supply is the same as the total value of demand. As the owners of the
farms, we are interested in transporting the product in a way to avoid
excess supply and meet all demands. We can deal with this issue in two
ways: do the transport on our own or find a logistics firm which will
do it for us. The Kantorovich lifting relies on the latter
perspective. We assume that for each place it sets up a price for the
logistics company at which it will buy a unit of our product (at farms
with overproduction) or sell it (at cafés with excess demand). This is
equivalent to giving a function $f \colon X \to [0, \infty]$. We
require that prices are \emph{competitive}, that is for all
$x,y \in X$, we have that $|f(x) - f(y)| \leq d(x,y)$, which is
equivalent to saying that $f$ is a non-expansive function from $d$
into the Euclidean distance $d_e$ on $[0, \infty]$. Given a pricing
$f$, the profit made by the transportation company is given by
$c_f = \sum_{x \in X} f(x)P(x) - \sum_{x \in X}f(x)Q(X)= \sum_{x\in X}
f(x) \left(P(x) - Q(x)\right)$. Because the transportation company
wants to make the most profit, it is aiming to pick a pricing $f$
maximizing the formula given above. Combining all the moving parts
together we are left with formula
$d^\uparrow(P,Q) = \max \{ \sum_{x \in {X}} f(x)
\left(P(X)-Q(X)\right) \mid f \colon (X, d_X) \to ([0,\infty],
d_{e})\}$ defining the Kantorovich distance between probability
distributions $P$ and $Q$.

It is not straightforward to see, but in this example an optimal pricing
function is $f(A)=0$, $f(B)=3$, $f(C)=5$, which can be easily seen to be
non-expansive and yields a distance of $2.1$.

More abstractly, the formula above can be dissected into three pieces:
\begin{enumerate}
\item \label{it:transport-1}Taking all pricing plans $f$, which are
  non-expansive with respect to the Euclidean distance on
  $[0,\infty]$.
\item \label{it:transport-2}Evaluating each of the pricing plans, by calculating the
  expected value of $f$ given a distribution on the set $X$.
\item \label{it:transport-3}Picking a pricing plan which maximizes the
  difference between the expected values.
\end{enumerate}

In the following, we will concentrate on the directed case, where
distance functions can be asymmetric and $d_e$ is the directed
Euclidean distance, that is, $d_e(x,y) = y \ominus x$.

In the category-theoretic
generalization~\cite{bbkk:coalgebraic-behavioral-metrics,bgkm:hennessy-milner-galois}
the calculation of the expected value (step 2) is replaced with (the set
of) evaluation functions, intuitively scoring or testing the
observable behaviour. At the same time, the steps of taking all non-expansive plans for a given metric
and of generating a metric that maximizes the difference between expected
values (steps 1 and 3) generalize to the setting of quantales and
their residuation. In the next section, we show that the generalizations of those two steps
form a fibred adjunction.

\section{Setting Up a Fibred Adjunction}
\label{sec:adjunction}

One of the key aspects of this paper is equipping sets of states of
coalgebras with an extra structure of conformances (in particular
preorders or hemimetrics). The very idea of ``extra structure''
can be phrased formally through the lenses of fibrations and fibred
category theory, extending the ideas
of~\cite{bgkmfsw:logics-coalgebra-adjunction}. In particular, we show
that those results can be strengthened to the setting of fibred
adjunctions.



\subparagraph*{The category $\Vspred$.}
\label{sec:Vspred}
The adjunction-based framework
from~\cite{bgkmfsw:logics-coalgebra-adjunction}, besides working with
$\VV$-graphs, makes use of sets of $\VV$-valued predicates, i.e.,
maps of the form $p \colon X \to \VV$. We will use $\Vspred_X$ to
denote the collection of sets of $\VV$-valued predicates over some set
$X$. A morphism between sets $S \subseteq \VV^X$ and
$T \subseteq \VV^Y$ is a function $f \colon X \to Y$ satisfying
$f^{\bullet}(T) := \{q\circ f\mid q\in T\}\subseteq S$, where
$f^\bullet$ describes reindexing. We obtain a category $\Vspred$
with objects being pairs $(X, S \subseteq \VV^X)$ and arrows are
defined as above.

\subparagraph*{Grothendieck completion.} It turns out that both
$\Vgraph / \Vcat$ and $\Vspred$ can be viewed as fibred categories~\cite{Lawvere73,bkp:up-to-behavioural-metrics-fibrations-journal}. We
here only consider fibred categories arising from the Grothendieck
construction, which can be viewed equivalently as a special kind of
split indexed categories, that in our case are functors
$A : \Set^\op \to \Pos$. Intuitively, such functors assign to each set
a poset of ``extra structure'' and take functions $f : X \to Y$ to
monotone maps canonically reindexing the ``extra structure'' on set
$Y$ by $f$.

The \emph{Grothendieck
  completion}~\cite{j:categorical-logic-type-theory} of a functor
$A \colon \Set^\op\to \Pos$ is a category denoted $\int A$, whose
objects are pairs $(X,i)$ where $X\in \Set$ and $i\in A(X)$. The
arrows $(X,i)\to (Y,j)$ are maps $f\colon X\to Y$ such that
$i\vle A(f)(j)$, where $\vle$ is the partial order on $A(X)$. The
corresponding fibration is given by the forgetful functor
$U \colon \int A \to \Set$ taking each pair $(X,i)$ to its underlying
set $X$. The fibre of each set $X$ corresponds to the poset $A(X)$.

\subparagraph*{$\Vgraph / \Vcat$ and $\Vspred$ as Grothendieck
  completions.} The category $\Vgraph$ arises as the Grothendieck
completion of the functor $\Phi : \Set^{\op} \to \Pos$ that takes each
set $X$ to $\Phi(X)=(\Vgraph_X, \vle)$, the lattice of
$\VV$-valued relations equipped with the pointwise order
$\vle$. Given a function $f \colon X \to Y$, we define
$\Phi(f) = f^*$ by reindexing, where $f^*(d_Y) = d_Y\circ(f\times
f)$. Analogously for $\Vcat$.

Similarly, to obtain $\Vspred$, we define a functor
$\Psi : \Set^{\op} \to \Pos$ that maps each set $X$ to
$\Psi(X)=(\Vspred_X, \supseteq)$, the lattice of collections of
$\VV$-valued predicates on $X$ ordered by reverse
inclusion. Furthermore $\Psi(f) = f^\bullet$.



\begin{toappendix}
  \begin{lemma}
    $\Vgraph = \int \Phi$ and $\Vspred = \int \Psi$.
  \end{lemma}
  \begin{proof}
    The objects of $\int \Phi$ are pairs $(X, d_X)$, where $X$ is a
    set and $d_X$ is an element of the lattice $\Vgraph_X$ of
    $\VV$-valued relations. The arrows $f \colon (X, d_X) \to (Y,d_Y)$
    are functions $f \colon X \to Y$, which satisfy
    $d_X \vle f^*(d_Y)$, which is the same as
    $d_X \vle d_Y \circ (f \times f)$.  This precisely
    corresponds to the definition of the category $\Vgraph$.
	
    Similarly, the objects of $\int \Psi$ are pairs $(X, S)$, where
    $X$ is a set and $S \in \Vspred_X$ (and hence $S \subseteq \VV^X$)
    is the element of the lattice of subsets of $\VV^X$ ordered by
    reverse inclusion. The arrows $(X,S) \to (Y,T)$ are functions
    $f \colon X \to Y$, such that $S \supseteq f^\bullet(T)$, which is
    the same as $\{q \circ f \mid q \in T\} \subseteq S$. Again this
    precisely corresponds to the definition of the category $\Vspred$.
  \end{proof}
  Additionally, observe that if we were to restrict $\Phi$ to send
  each set $X$ to the lattice $(\Vcat_X, \vle)$ of
  $\VV$-categories on $X$, which is a sublattice of the lattice
  $(\Vgraph_X, \vle)$ of all $\VV$-graphs, then the
  Grothendieck completion would yield $\Vcat$ instead of $\Vgraph$.
\end{toappendix}
\subparagraph*{Galois connection on the fibres.} In the
adjunction-based framework from
\cite{bgkmfsw:logics-coalgebra-adjunction}, the Kantorovich lifting of
a functor is phrased through the means of a contravariant Galois
connection between the fibres of $\Vcat$ and $\Vspred$ and we here
generalize from $\Vcat$ to $\Vgraph$. For each set $X$, we define a map
$\alpha_X \colon \Vspred_X \to \Vgraph_X$ given by:
\[\alpha_X(S)(x_1,x_2) = \vMeet\nolimits_{f\in S} d_\VV(f(x_1),f(x_2)) \qquad (S \subseteq \VV^X)\] 
Intuitively, $\alpha_X$ takes a collection $S$ of $\VV$-valued predicates on $X$ and generates the greatest conformance on $X$ that turns all predicates into non-expansive maps. For the other part of the Galois connection, we have a map $\gamma_X \colon \Vgraph_X \to \Vspred_X$ defined by the following:
      \[ \gamma_X(d_X) = \{f\colon X\to \VV \mid d_\VV \circ (f \times f ) \vge d_X\} \]
      Given a conformance $d_X \colon X \times X \to \VV$, $\gamma_X$ generates a set of $\VV$-valued predicates on $X$ which are non-expansive maps from $d_X$ to the residuation distance. 
      As mentioned before, we can instantiate the previous result \cite[Theorem~7]{bgkmfsw:logics-coalgebra-adjunction} and obtain the following:
\begin{theorem}[\cite{bgkmfsw:logics-coalgebra-adjunction}]\label{thm:galois_connection}
  Let $X$ be an arbitrary set, $d_X \colon X \times X \to \VV$ a
  $\VV$-graph and $S \subseteq \VV^X$ a collection of $\VV$-valued
  predicates. Then $\alpha_X$ and $\gamma_X$ are both antitone (in
  $\subseteq$, $\vle$) and form a contravariant Galois connection:
  \[
    d_X \vle \alpha_X (S) \iff S \subseteq \gamma_X(d_X).
  \] 
\end{theorem}

\begin{example}
  In the setting of~\Cref{sec:transport}, $\gamma$ corresponds to
  Step~\ref{it:transport-1} and $\alpha$ to
  Step~\ref{it:transport-3}.
\end{example}

\subparagraph*{Fibred Adjunction.} \Cref{thm:galois_connection} is a
``local'' property that only holds fiberwise. However, it turns out
that we can argue something stronger, namely that we have a fibred
adjunction situation. This is a ``global'' property, as fibred
functors additionally preserve the notion of reindexing. Note that
every natural transformation between functors of type
$\Set^{\op} \to \Pos$ (a so-called morphism of split indexed categories)
\cite[Definition~1.10.5]{j:categorical-logic-type-theory} induces a
fibred functor between the corresponding Grothendieck completions.

One can quite easily verify that $\alpha$ is natural on $\Vgraph$, while
$\gamma$ is only laxly natural on $\Vgraph$ and natural on $\Vcat$. For
the latter result, we rely on a quantalic version of the
McShane-Whitney extension theorem
\cite{McS:ExtensionOfFunctions,W:ExtensionOfFunctions}, mentioned also
in \cite{Lawvere73} for the real-valued case.

\begin{toappendix}
  \begin{lemma}
    \label{lem:alpha-natural}
    $\alpha\colon\Psi\Rightarrow\Phi$ is natural.
  \end{lemma}
  \begin{proof}
    Let $f \colon X \to Y$ be an arbitrary function. We need to check
    the commutativity of the following diagram in $\Pos$:
    \[\begin{tikzcd}
	{\Psi(Y)} && {\Psi(X)} \\
	{\Phi(Y)} && {\Phi(X)} \arrow["{\alpha_X}", from=1-3, to=2-3]
        \arrow["{\alpha_Y}", from=1-1, to=2-1] \arrow["{f^\bullet}",
        from=1-1, to=1-3] \arrow["{f^*}", from=2-1, to=2-3]
      \end{tikzcd}\] Note that as a consequence of
    \Cref{thm:galois_connection}, we have antitonicity of $\alpha_X$ and
    hence $\alpha_X$ is a monotone map from the lattice of sets of
    predicates with inverse inclusion order to the lattice of
    $\VV$-valued relations. It suffices to check
    $\alpha_X \circ f^\bullet = f^* \circ \alpha_Y$. Let
    $T \subseteq \VV^Y$ and $x_1, x_2 \in X$ and consider the
    following:
	\begin{align*}
		(\alpha_X \circ f^\bullet)(T)(x_1,x_2) &= \vMeet_{p \in T} d_\VV ((p \circ f)(x_1), (p \circ f)(x_2)) \\
		&=  \alpha_Y (T) (f(x_1), f(x_2)) \\
		&= (f^* \circ \alpha_Y(T))(x_1, x_2) \qedhere
	\end{align*}

\end{proof}
Unfortunately, in the case of $\gamma$, without any extra assumptions, we can only prove a weaker statement, namely that it is laxly natural.
\begin{lemma}\label{lem:gamma_laxly_natural}
	$\gamma \colon \Phi \Rightarrow \Psi $ is laxly natural, that is for all functions $f \colon X \to Y$ and all $\VV$-valued relations $d_Y$ on the set $Y$ we have that \((f^\bullet \circ \gamma_Y)(d_Y) \subseteq (\gamma_X \circ f^*)(d_Y)\). 
\end{lemma}
\begin{proof}
Observing that $\gamma_X$ is antitone, we can show that for all sets $X$, $\gamma_X$ is a monotone function from the lattice of $\VV$-valued relations to the lattice of collections of $\VV$-valued predicates (ordered by reverse inclusion order). The only thing we need to show is that for all functions $f \colon X \to Y$, we have that $(f^\bullet \circ \gamma_Y)(d_Y) \subseteq (\gamma_X \circ f^*)(d_Y)$ for all $\VV$-valued relations $d_Y$ on the set $Y$. 

Assume that $p \colon X \to \VV$ belongs to $(f^\bullet \circ \gamma_Y)(d_Y)$. In other words, $p = q \circ f $ for some $q \colon Y \to \VV$, such that $d_{\VV} \circ (q \times q) \vge d_Y$. Observe that pre-composing $f \times f$ yields the inequality $d_\VV \circ ((q \circ f) \times (q \circ f)) \vge d_Y \circ (f \times f)$ and therefore $d_{\VV} \circ (p \times p) \vge d_Y \circ (f \times f) = f^{*}(d_Y)$. Hence, $p \in (\gamma_X \circ f^*)(d_X)$, which completes the proof.
\end{proof}
To obtain the converse inclusion, we need to be able to factorize any
$\VV$-valued predicate on $X$, which is non-expansive with respect to
$f^*(d_Y)$ as the composition of $f$ and a $\VV$-valued predicate on $Y$ which is non-expansive with respect to $d_Y$. Since $f$ can be seen as isometry between $(X,f^*(d_Y))$ and $(Y,d_Y)$, we are essentially focusing on extending a non-expansive map into a residuation metric through an isometry, which in the case of pseudometric spaces can be proved through the McShane-Whitney Extension Theorem. It turns out that those results can be generalized to the setting of $\VV$-categories.


\end{toappendix}

\begin{toappendix}
\begin{theoremrep}[Quantalic McShane-Whitney Extension Theorem]\label{thm:extension_theorem}
Let $\VV$ be a quantale and $(X, d)$ be an object of $\Vcat$. For any
subset $Y\subseteq X$ and any non-expansive map $f\colon (Y, d) \to (\VV, d_\VV)$ there is a non-expansive extension $\bar{f}\colon (X, d) \to (\VV, d_\VV)$ of $f$.
\end{theoremrep}

\begin{proof}
Define for every element $y\in Y$ a function $g_y\colon X \to \VV$ by
\[g_y(x) = d_\VV(d(x, y), f(y))\]
Note that all the functions $g_y$ are non-expansive, since for $x_1,
x_2 \in X$ we have $d(x_1, x_2) \otimes d(x_2, y)\vle d(x_1,
y)$ (triangle inequality) and hence
\begin{align*}
  d(x_1, x_2)
  &\vle d_\VV(d(x_2, y), d(x_1, y)) \\
  &\vle d_\VV(d_\VV(d(x_1, y), f(y)), d_\VV(d(x_2, y), f(y))) \tag{\Cref{lem:residuation}\eqref{it:residuation-9}} \\
  &= d_\VV(g_y(x_1), g_y(x_2)).
\end{align*}

Now define $\bar{f}\colon X \to \VV$ by
\[\bar{f}(x) = \vMeet_{y\in Y} g_y(x) = \vMeet_{y\in Y} d_\VV(d(x, y), f(y))\]
Once can easily verify that $\bar{f}$ is also non-expansive, because
\begin{align*}
d_\VV(\bar{f}(x_1), \bar{f}(x_2)) &= d_\VV\Big( \vMeet_{y\in Y} g_y(x_1), \vMeet_{y\in Y} g_y(x_2) \Big)\\
&\vge \vMeet_{y\in Y} d_\VV(g_y(x_1), g_y(x_2))  \tag{\Cref{lem:residuation}\eqref{it:residuation-10}}\\
&\vge \vMeet_{y\in Y} d(x_1, x_2) = d(x_1, x_2)
\end{align*}
It remains to show that $\bar{f}(y) = f(y)$ for all $y\in Y$. Consider
$z\in Y$, then from the antitonicity of $d_\VV$ in its first argument
and \Cref{lem:residuation}\eqref{it:residuation-4}:
\[\bar{f}(z) \vle  d_\VV(d(z, z), f(z)) \vle d_\VV(k, f(z)) = f(z)\]
On the other hand, by the non-expansiveness of $f$, for all $y\in Y$ we
have by commutativity:
\begin{align*}
d_\VV(f(z), f(y)) \vge d(z, y) &\iff d(z, y) \otimes f(z) \vle f(y)\\
&\iff d_\VV(d(z, y), f(y)) \vge f(z)
\end{align*}
And therefore
\[\bar{f}(z) = \vMeet_{y\in Y} d_\VV(d(z, y), f(y)) \vge f(z) \qedhere \]

\end{proof}

\begin{remark}
  Note that the function $\bar{f}$ constructed in the proof is the largest non-expansive extension of $f$.
That means for all non-expansive extensions $h\colon (X, d) \to (\VV, d_\VV)$ of $f$ we have $h \vle \bar{f}$.
For if $h$ is another non-expansive extension of $f$, then for all $x\in X$ and all $y\in Y$ we have
\begin{align*}
d_\VV(h(x), f(y)) = d_\VV(h(x), h(y)) \vge d(x, y) &\iff h(x) \otimes d(x, y) \vle f(y)\\
&\iff h(x) \vle d_\VV(d(x, y), f(y))
\end{align*}
and therefore
\[h(x) \vle \vMeet_{y\in Y} d_\VV(d(x, y), f(y)) = \bar{f}(x).\]
Similarly, the function $g\colon X \to \VV$ defined by
\[g(x) = \vJoin_{y\in Y} f(y) \otimes d(y, x)\]
is the smallest non-expansive function extending $f$.
\end{remark}
As a corollary, we obtain the desired property needed to show the converse inclusion.
\end{toappendix}

\begin{theoremrep}
  \label{thm:vrel-extension}
  Let $d_X\in\Vcat_X$ and $d_Y\in \Vcat_Y$ be elements of $\Vcat$. If
  $i \colon (Y, d_Y) \to (X, d_X)$ is an isometry, then for any
  non-expansive map $f \colon (Y,d_Y) \to (\VV, d_\VV)$ there exists a
  non-expansive map $g \colon (X,d_X) \to (\VV, d_{\VV})$ such that
  $f = g \circ i $.
\end{theoremrep}

\begin{proof}
	The function $i \colon Y \to X$ can be factorized as the composition of $e \colon Y \to i[Y]$, where $i[Y]$ denotes the image of $Y$ under $i$ and $m \colon i[Y] \to X$, where $m$ is the canonical inclusion function of $i[Y]$ into $X$. We can equip $i[Y]$ with a distance $d_{i[Y]} \colon i[Y] \times i[Y] \to \VV$ defined to be the restriction of $d_X$ to the domain $i[Y] \times i[Y]$. One can observe that $(i[Y],d_{i[Y]})$ is also an element of $\Vcat$ and that $m \colon (Y, d_Y) \to (i[Y], d_{i[Y]})$ is an isometry as a consequence of $i$ being an isometry.
	
	We define a function $h \colon i[Y] \to \VV$ to be given by $h(e(y)) = f(y)$. We now argue that $h$ is well-defined. Let $y_1, y_2 \in Y$, such that $e(y_1)=e(y_2)$. Since $f \colon (Y,d_Y) \to (\VV, d_{\VV})$ is non-expansive, we have that:
	$$
		d_{\VV}(f(y_1), f(y_2)) \vge d_Y(y_1, y_2) = d_{i[Y]} (e(y_1), e(y_2)) \vge k
	$$
	The second to last step relies on the fact that $d_{i[Y]}$ is reflexive. By the definition of $d_{\VV}$ we have that the above inequality is equivalent to 
	$$
	k \otimes f(y_1) \vle f(y_2) \iff f(y_1) \vle f(y_2)
	$$
	The second step above relies on the fact that $k$ is the unit of $\otimes$. The argument that $f(y_2) \vle f(y_1)$ is symmetric and therefore omitted. 
	
	We now argue that $h \colon (i[Y], d_{i[Y]}) \to (\VV, d_\VV)$ is non-expansive. Let $x_1, x_2 \in i[Y]$. Since $e$ is surjective, there must exist $y_1, y_2$ with $e(y_1) = x_1$ and $e(y_2) = x_2$.
	\begin{align*}
		d_{\VV}(h(x_1),h(x_2)) &= d_{\VV}(h(e(y_1)), h(e(y_2))) \tag{$e$ is surjective} \\
		&= d_{\VV}(f(y_1), f(y_2)) \tag{Def. of $h$} \\
		&\vge d_Y (y_1, y_2) \tag{$f$ is non-expansive}\\
		&= d_{i[Y]} (e(y_1), e(y_2)) \tag{$e$ is an isometry} \\
		&= d_{i[Y]} (x_1, x_2)
	\end{align*}
	
	We can sum up the situation as the commuting diagram in $\Vcat$.
\[\begin{tikzcd}
	{(Y,d_Y)} \\
	{(i[Y], d_{i[Y]})} && {(\VV, d_{\VV})} \\
	{(X, d_X)}
	\arrow["e", from=1-1, to=2-1]
	\arrow["m", hook', from=2-1, to=3-1]
	\arrow["h", from=2-1, to=2-3]
	\arrow["f", from=1-1, to=2-3]
	\arrow["g", from=3-1, to=2-3]
      \end{tikzcd}\] The upper triangle commutes by the definition of
    $h$, while $g$ exists and makes the bottom one commute because of
    \Cref{thm:extension_theorem}. Since $i = m \circ e$, we
    have that there exists non-expansive $g$, such that
    $g \circ i = f$, which completes the proof.
\end{proof}

\begin{toappendix}
Note, that the property above does not hold for arbitrary $\VV$-graphs and only works for $\VV$-categories. To make use of that, we restrict $\Phi(X)$ to be the lattice $(\Vcat_X, \vle)$ of $\VV$-categories on $X$.
\begin{lemma}\label{lem:gamma_natural}
	When restricted to $\Vcat$, $\gamma \colon \Phi \to \Psi$ is natural.
\end{lemma}
\begin{proof}
It suffices to show the converse inclusion to the one argued in \Cref{lem:gamma_laxly_natural}. Let $p \colon X \to \VV$ be an arbitrary $\VV$-valued predicate on $X$, such that $d_Y \circ (f \times f) \vle d_\VV \circ (p \times p)$. We want to show that there exists a non-expansive map $q \colon (Y,d_Y) \to (\VV, d_\VV)$, such that $p = q \circ f$. 
Observe that $p \colon X \to \VV$ is a non-expansive map from $(X, d_Y \circ (f \times f))$ to $(\VV, d_\VV)$ and $f \colon (Y, d_Y \circ (f \times f)) \to (Y, d_Y)$ is an isometry.
\[\begin{tikzcd}
	{(X, d_{Y} \circ (f \times f))} \\
	{(Y,d_Y)} && {(\VV, d_{\VV})}
	\arrow["f", from=1-1, to=2-1]
	\arrow["{\exists q}", from=2-1, to=2-3]
	\arrow["p", from=1-1, to=2-3]
\end{tikzcd}\]
We can use \Cref{thm:vrel-extension} and conclude that there exists a non-expansive function $q \colon (Y, d_Y) \to (\VV, d_\VV)$ such that $p = q \circ f$. Hence, $p \in (f^\bullet \circ \gamma_Y)(d_Y)$, which completes the proof.
\end{proof}
\end{toappendix}

\begin{propositionrep}
  We have that $\alpha\colon\Psi\Rightarrow\Phi$ is natural and
  $\gamma \colon \Phi \Rightarrow \Psi $ is laxly natural, that is for
  all functions $f \colon X \to Y$ and all $\VV$-valued relations
  $d_Y$ on the set $Y$ we have that
  \((f^\bullet \circ \gamma_Y)(d_Y) \subseteq (\gamma_X \circ
  f^*)(d_Y)\).  When restricted to $\Vcat$,
  $\gamma \colon \Phi \Rightarrow
   \Psi$ is natural.
\end{propositionrep}

\begin{proof}
  Follows from
  Lemmas~\ref{lem:alpha-natural},~\ref{lem:gamma_laxly_natural}
  and~\ref{lem:gamma_natural}.
\end{proof}

Note that we can safely make this restriction, while still keeping
$\alpha$ to be well-defined, as its image always lies within $\Vcat$.
\begin{toappendix}
  \begin{lemma}
    \label{lem:image-alpha-vcat}
    For all sets $X$ and collections of predicates
    $S \subseteq \VV^X$, $\alpha_X(S)$ is a $\VV$-category.
  \end{lemma}
  \begin{proof}
    Let $x \in X$. For reflexivity, we have that for all $p \in S$,
    $d_{\VV}(p(x),p(x)) \vge k$ and hence
    \[
      \alpha_X(S)(x,x)=\vMeet_{p \in S} d_{\VV}(p(x), p(x))
      \vge k
    \]
    For transitivity, let $x_1, x_2 \in X$.
    We have to show the following for all $z \in X$:
    \[
      \alpha_X(S)(x_1, z) \otimes \alpha_X(S)(z, x_2) \vle
      \alpha_X(S)(x_1,x_2) = \vMeet_{p \in S} d_{\VV}(p(x_1), p(x_2))
    \]
    To show the above, it is enough to argue that for all $p \in S$,
    we have that
    $\alpha_X(S)(x_1,z) \otimes \alpha_X(z, x_2) \vle
    d_{\VV}(p(x_1), p(x_2))$. We have the following:
    \begin{align*}
      \alpha_X(S)(x_1, z) \otimes \alpha_X(S)(z, x_2) &= \left(\vMeet_{r \in S}d_{\VV}(r(x_1), r(z))\right) \otimes  \left(\vMeet_{r \in S}d_{\VV}(r(z), r(x_2))\right)\\
      &\vle  \vMeet_{r \in S} d_{\VV}(r(x_1), r(z)) \otimes d_{\VV}(r(z), r(x_2))\tag{$\left(\vMeet_{i \in I} a_i\right) \otimes \left(\vMeet_{i \in I} b_i\right) \vle \vMeet_{i \in I} a_i \otimes b_i$}\\
      &\vle \vMeet_{r \in S} d_{\VV}(r(x_1), r(x_2))
      \tag{$d_{\VV}$ is transitive, \Cref{lem:residuation}\eqref{it:residuation-3}}\\
      &\vle d_{\VV}(p(x_1), p(x_2)) \tag{$p \in S$}
    \end{align*}
    which completes the proof.
  \end{proof}
  Additionally, when restricting our attention to $\Vcat$ (instead of
  all of $\Vgraph$), it turns out that the co-closure operator of the
  Galois connection (\Cref{thm:galois_connection}) becomes an
  identity.
  \begin{lemma}
    \label{lem:coclosure-id}
    The co-closure $\alpha_X\circ \gamma_X$ is the identity, when
    restricted to $\Vcat$.
  \end{lemma}

  \begin{proof}

    Let $(X, d_X)$ be an object of $\Vcat$. The inequality
    $d_X \vle (\alpha_X \circ \gamma_X)(d_X)$ holds
    immediately because of \Cref{thm:galois_connection}. For the other
    direction, we will have to show that for all $x_1, x_2 \in X$, we
    have that
    $d_X(x_1,x_2) \vge (\alpha_X \circ
    \gamma_X)(d_X)(x_1,x_2)$.  To observe that, we define a function
    $p \colon X \to \VV$ given by $p(x)= d_X(x_1, x)$. First, recall
    that for all $x_1,x_2 \in X$
$$
\alpha_X (\gamma_X(d_X))(x_1, x_2) = \vMeet_{\substack{p \colon X \to
    \VV\\ d_X \vle d_{\VV} \circ (p \times p)}} d_{\VV}(p(x_1),
p(x_2))
$$
To see that $p$ is non-expansive with respect to $d_\VV$, note that by the definition of residuation
\begin{align*}
  d_X(x_1,x_2) \vle d_{\VV}(p(x_1), p(x_2))  
  &\iff
  p(x_1) \otimes d_X(x_1,x_2) \vle p(x_2) \\
  &\iff
  d_X(x_1,x_1) \otimes d_X(x_1,x_2) \vle d_X(x_1,x_2)
\end{align*}
and the latter follows by the triangle inequality (\Cref{lem:residuation}\eqref{it:residuation-3}).
Therefore, $p\in\gamma_X(d_X)$, and hence
\begin{align*}
  \alpha_X(\gamma_X(d_X))(x_1,x_2)
  &\vle d_{\VV}(p(x_1), p(x_2)) \\
  &=d_{\VV}(d_X(x_1,x_1), d_X(x_1,x_2)) \\
  &\vle d_{\VV}(k, d_X(x_1,x_2)) \tag{Antitonicity of $d_{\VV}(-,-)$ in the first argument} \\
  &=d_X(x_1,x_2) \tag{$x = d_{\VV}(k,x)$, \Cref{lem:residuation}\eqref{it:residuation-4}}
\end{align*}
\end{proof}
\end{toappendix}

\begin{propositionrep}
  For all sets $X$ and $S \subseteq \VV^X$, $\alpha_X(S)$ is a
  $\VV$-category, i.e., an object of $\Vcat$.  The co-closure
  $\alpha_X\circ \gamma_X$ is the identity when restricted to
  $\Vcat$. Combined, this implies that for $d\in\Vgraph_X$,
  $\alpha_X(\gamma_X(d))$ is the metric closure of $d$, i.e.~the
  least element of $\Vcat_X$ above $d$.
\end{propositionrep}

\begin{proof}
  The first two statements follow from Lemmas~\ref{lem:image-alpha-vcat}
  and~\ref{lem:coclosure-id}.

  Now let $d\in\Vgraph_X$, $d'\in \Vcat_X$ with $d\vle d'$. By
  antitonicity of $\alpha$, $\gamma$ it holds that
  $\alpha_X(\gamma_X(d))\vle \alpha_X(\gamma_X(d')) =
  d'$. Furthermore $\alpha_X(\gamma_X(d))$ is a $\VV$-category, which
  implies the statement.
\end{proof}

Because of the Grothendieck construction
(\cite[Theorem~1.10.7]{j:categorical-logic-type-theory}), $\alpha$ and
$\gamma$ respectively correspond to fibred functors
$\alpha \colon \Vspred \to \Vcat$ and $\gamma : \Vcat \to
\Vspred$. Both functors keep morphisms unchanged and act on objects by
respectively applying the appropriate components of $\alpha$ and
$\gamma$. Now, we can state the strengthened version of
\Cref{thm:galois_connection}, by showing that $\alpha$ and $\gamma$
form a fibred adjunction. Note that due to the choice of the
orderings, $\gamma$ becomes the left and $\alpha$ the right adjoint
(cf.~\cite{kwrk:composing-codensity-bisimulations}).
\begin{theoremrep}\label{thm:adjunction}
  There is an adjunction $\gamma \dashv \alpha$.
\end{theoremrep}
\begin{proof}

  
Let $d_X\in\Vcat_X$ be an arbitrary reflexive and transitive $\VV$-valued relation on $X$, $T \subseteq \VV^Y$ and $f \colon X \to Y$ a non-expansive map from $(X,d_X)$ to $(Y,\alpha_Y(T))$. We will define a non-expansive map $\eta_{d_X} \colon (X, d_X) \to (X,\alpha_X(\gamma_X(d_X)))$ and show that there exists a unique arrow $g \colon \gamma_X(d_X) \to T$ which makes the following diagram commute.
\[\begin{tikzcd}
	{(X,d_X)} && (X,\alpha_X(\gamma_X(d_X))) \\
	\\
	{(Y,\alpha_Y(T))}
	\arrow["f", from=1-1, to=3-1]
	\arrow["{\eta_{d_X}}", from=1-1, to=1-3]
	\arrow["{\alpha (g)}", dashed, from=1-3, to=3-1]
\end{tikzcd}\]
First, we set $\eta_{d_X} = \id$ to be the identity map. In order to show that it is non-expansive, we need to argue that $d_X \vle \alpha_X(\gamma_X (d_X)) $. This immediately follows from \Cref{thm:galois_connection}, where we have shown that $\alpha_X$ and $\gamma_X$ form a contravariant Galois connection.

Recall that $\alpha(g) = g$ and hence the diagram above would uniquely commute if we set $g = f$. The only remaining thing is to argue that $f$ is a morphism in $\Vspred$ between $(X, \gamma_X(d_X))$ and $(Y, T)$. In other words, we need to prove that:
$$
f^\bullet(T) = \{q \circ f \mid q \in T\} \subseteq \gamma_X(d_X) = \{p \colon X \to \VV \mid d_X \vle d_{\VV} \circ (p \times p)\}
$$
Because $f \colon (X, d_X) \to \alpha_X (T)$ is non-expansive, we have that for all $x_1, x_2 \in X$:
$$d_X(x_1,x_2) \vle \vMeet_{q	 \in T} d_{\VV}(q(f(x_1)), q(f(x_2)))$$ 
and therefore for any $q \in T$, we have that:
$$
d_X(x_1,x_2) \vle d_{\VV} ((q \circ f) (x_1), (q \circ f) (x_2))
$$
Observe that because of the inequality above every element of $f^{\bullet}(T)$ also belongs to $\gamma_X(d_X)$, which completes the proof.
\end{proof}

\section{The Coalgebraic Kantorovich Lifting}
\label{sec:kantorovich}

\subsection{Definition of the Coalgebraic Kantorovich Lifting}
\label{sec:def-kantorovich}

\todo{Add small comment on parallel to ref. [2]. B: done}The
coalgebraic Kantorovich lifting
\cite{bbkk:coalgebraic-behavioral-metrics} (originally defined for the
real-valued case and for a single evaluation map) -- extended to
codensity liftings in \cite{ks:codensity-liftings-monads} -- is
parametric in a set of \emph{evaluation functions} for a set
functor~$F$.  Evaluation functions are maps of type $F\VV\to\VV$ (generalizing the expected value
computation in the traditional Kantorovich lifting) and as such can be used to lift
$\VV$-valued predicates on a set $X$ to $\VV$-valued predicates on the
set $FX$.  More precisely, given an evaluation function
$\ev\colon F\VV\to\VV$ and a predicate $f\colon X\to\VV$ we obtain the
predicate $\ev\circ Ff\colon FX\to\VV$.  This operation extends to
sets of evaluation functions and sets of $\VV$-valued predicates,
where a set $\Lambda^F$ of evaluation functions for $F$ induces the
fibred functor $\Lambda^F\colon \Vspred\to\Vspred$, defined on
$S\subseteq\VV^X$ as follows:
\begin{equation*}
  \Lambda^F_X(S) = \{\ev\circ Ff \mid \ev\in\Lambda^F, f\in S \} \subseteq \VV^{FX}.
\end{equation*}
The Kantorovich lifting can now be restated via the fibred adjunction introduced previously.
Given $F$ and $\Lambda^F$ as above, we can define its Kantorovich
lifting $\kant{\Lambda^F}$ as follows:
\[ \kant{\Lambda^F} = \alpha{F} \circ \Lambda^F \circ \gamma \] or
more concretely, for an object $d_X$ of $\Vgraph$ and
$s,t\in FX$:
\[\kant{\Lambda^F}(d_X)(s,t) = \bigsqcap\nolimits_{\ev\in\Lambda^F}
  \bigsqcap\nolimits_{f\in\gamma_X(d_X)} d_\VV(\ev(Ff(s)),\ev(Ff(t))). \] If
the set $\Lambda^F$ is clear from the context we sometimes write
$\overline{F}$ instead of $\kant{\Lambda^F}$.  

\todo{Note on whether there is a canonical lifting. B: I saw that
  Wojtek already did this.}
\begin{lemmarep}
  The Kantorovich lifting of a functor $F : \Set \to \Set$ is a
  functor $\overline{F} = \kant{\Lambda^F} \colon \Vgraph \to \Vgraph$, and fibred when restricted to $\Vcat$.
\end{lemmarep}
\begin{proof}
	We need to show that for an arbitrary non-expansive function $f \colon (X, d_X) \to (Y,d_Y)$, $Ff$ becomes a non-expansive function between $\overline{F}(X,d_X)$ and $\overline{F}(Y,d_Y)$. In other words, we require that:
	\[
		(\alpha_{FX} \circ \Lambda_X^F \circ \gamma_X)(d_X)\vle (\alpha_{FY} \circ \Lambda_Y^F \circ \gamma_Y)(d_Y) \circ (Ff \times Ff)
	\]
	First, we simplify the right-hand side:
	\begin{align*}
		&\left((\alpha_{FY} \circ \Lambda_Y^F \circ \gamma_Y)(d_Y) \circ (Ff \times Ff)\right)\\
		=&\left((\alpha_{FY} \circ \Lambda_Y^F \circ \gamma_Y(d_Y) )\circ (Ff \times Ff)\right)\\
		=&\left(Ff^* \circ \alpha_{FY} \circ \Lambda_Y^F \circ \gamma_Y(d_Y) \right) \tag{Definition of $(-)^*$} \\
		= &\left(\alpha_{FX} \circ (Ff)^\bullet \circ \Lambda_Y^F \circ \gamma_Y (d_Y)\right) \tag{$\alpha$ is natural}\\
		= &\left(\alpha_{FX}  \circ \Lambda_Z^F \circ f^\bullet \circ \gamma_Y (d_Y)\right)
	\end{align*}
	Hence, we have to show the following:
	\[
		(\alpha_{FX} \circ \Lambda_X^F \circ \gamma_X)(d_X)  \vle \left(\alpha_{FX} \circ \Lambda_X^F \circ f^\bullet \circ \gamma_Y (d_Y)\right)
	\]
	Since $\alpha_{FX}$ is antitone (\Cref{thm:galois_connection}), it suffices to argue the following:
	\[
	(\Lambda_X^F \circ \gamma_X)(d_X) \supseteq (\Lambda_X^F \circ f^\bullet \circ \gamma_Y)(d_Y)
	\]
	Because of the lax naturality of $\gamma$ (\Cref{lem:gamma_laxly_natural}), we simplify the above even further:
	\[
		(\Lambda_X^F \circ \gamma_X)(d_X) \supseteq (\Lambda_X^F \circ \gamma_X \circ f^*)(d_Y)	
	\]
	Since it is not hard to see that $\Lambda_X^F$ is monotone, we are left with proving:
	\[
	\gamma_Y (d_Y \circ (f \times f)) \subseteq \gamma_X(d_X)
	\]
	Let $p \in \gamma_X (d_Y \circ (f \times f))$. In other words, for all $x_1, x_2 \in X$ we have that \[d_{\VV}(p(x_1), p(x_2)) \vge d_Y(f(x_1), f(x_2))\]
	Combining the above with the fact that $f$ is a non-expansive function from $(X,d_X)$ to $(Y, d_Y)$, we obtain that $d_{\VV}(p(x_1), p(x_2)) \vge d_X(x_1, x_2)$
	and therefore $p \in \gamma_X(d_X)$, thus showing functoriality.

 For fibredness, consider a function $f:X\to Y$ in $\Set$, and an object $(Y,d_Y)$ in $\Vcat$. We wish to see that for all $s,t\in FX$ \[\kant{\Lambda^F}(d_Y\circ (f\times f))(s,t)=\kant{\Lambda^F}(d_Y)\circ (Ff\times Ff)(s,t)\]  This is equivalent to
 \begin{eqnarray*}
   && \bigsqcap_{\ev\in\Lambda^F} \bigsqcap_{p\in\gamma_X(d_Y\circ
     (f\times f))} d_\VV(\ev(Fp(s)),\ev(Fp(t))) \\ & = &
   \bigsqcap_{\ev\in\Lambda^F} \bigsqcap_{q\in\gamma_Y(d_Y)}
   d_\VV(\ev(Fq(Ff(s))),\ev(Fq(Ff(t))))
\end{eqnarray*}
and since in the category $\Vcat$, we have that $(f^\bullet \circ \gamma_Y)(d_Y) = (\gamma_X \circ f^*)(d_Y)$, then the above equality holds, completing the proof. \qedhere
	
\end{proof}
\begin{remark}
	Instantiating the construction above with the distribution
        functor $\mathcal{D}$ and a single evaluation function
        $\expect{}$ taking the expected values yields the usual
        Kantorovich lifting, while in the case of powerset functor
        $\mathcal{P}$ and a single evaluation function $\sup$, one
        obtains the  (directed) Hausdorff metric. Despite those two
        instantiations corresponding to well-known constructions,
        there is no well-defined notion of \emph{canonical lifting}
        and there are often different possibilities for a given functor.
        For example, the usual symmetric Hausdorff distance arises
        by additionally considering the dual evaluation
        function~$\inf$~\cite{ws:hemimetrics-fuzzy-lax}.
        We will later also see that constant factors admit more than
        one choice of evaluation functions.
	\end{remark}
\subsection{Compositionality of the Kantorovich Lifting}
\label{sec:compositionality}

When an endofunctor is given as the composition of two or more
individual set functors, it is natural to ask under which conditions
its Kantorovich lifting is also the composition of Kantorovich
liftings of the respective functors.  Specifically, our aim in this
section is to identify situations where this composition happens
already at the level of the underlying sets of evaluation maps.  If
$\ev_F$ is an evaluation function for $F$ and $\ev_G$ is an evaluation
function for $G$, then an evaluation function for $FG$ is given by
$\ev_F * \ev_G := \ev_F \circ F\ev_G$.  Extending this to sets of
evaluation functions, we put
\begin{math}
  \Lambda^F * \Lambda^G = \{\ev_F * \ev_G \mid \ev_F\in\Lambda^F, \ev_G\in\Lambda^G\}
\end{math}
for sets $\Lambda^F$ and $\Lambda^G$ of evaluation functions for functors $F$ and $G$, respectively.

\begin{definition}[Compositionality]
  \label{defn:compositionality}
  Given two functors $F$ and $G$ and sets $\Lambda^F$ and $\Lambda^G$ of evaluation functions, we say that we have \emph{compositionality} if
  \begin{math}
    \kant{\Lambda^F} \circ \kant{\Lambda^G} = \kant{\Lambda^F * \Lambda^G}.
  \end{math}
\end{definition}
%
%
Expanding definitions, compositionality amounts to showing that 
\begin{equation}
  \label{eqn:compositional-unfolded}
  \kant{\Lambda^F} \circ \kant{\Lambda^G} = \alpha FG \circ
  \Lambda^F \circ \gamma G \circ \alpha G \circ \Lambda^G \circ \gamma
  = \alpha FG \circ \Lambda^F \circ \Lambda^G \circ \gamma
  = \kant{\Lambda^F * \Lambda^G}.
\end{equation}
One inequality (`$\vle$') always holds:
As $\alpha$ and $\gamma$ form a Galois connection \cite{bgkmfsw:logics-coalgebra-adjunction}, we have $\id_{\Vspred_X} \subseteq \gamma_{GX} \circ \alpha_{GX}$, and thus we may use antitonicity of $\alpha$ to deduce `$\vle$' in~\eqref{eqn:compositional-unfolded}.
Baldan et al.~\cite[Lemma~7.5]{bbkk:coalgebraic-behavioral-metrics} prove this for the special case of pseudometric liftings.


The other inequality, `$\vge$', does not hold in general, and
requires more work.  Still, one notices that a sufficient
condition is that
$\Lambda^F\circ \gamma \circ \alpha \subseteq \gamma{F} \circ
\alpha{F}\circ\Lambda^F$:
\begin{eqnarray*}
  \kant{\Lambda^F} \circ \kant{\Lambda^G} &
  = & \alpha{FG}\circ\Lambda^F\circ \gamma{G} \circ
  \alpha{G}\circ\Lambda^G\circ\gamma \\
  & \vge & \alpha{FG} \circ \gamma{FG} \circ \alpha{FG} \circ \Lambda^F\circ\Lambda^G\circ\gamma
  =\alpha FG\circ\Lambda^{FG}\circ\gamma = \kant{\Lambda^F * \Lambda^G},
\end{eqnarray*}
using that $\alpha\circ \gamma\circ \alpha = \alpha$ for every Galois
connection. Note that it is enough to prove the sufficient condition
on non-empty sets, since $\gamma$ always yields a non-empty set.

Before discussing the problem of systematically constructing sets of
evaluation functions such that the sufficient condition (lax
commutativity of $\Lambda^F$ with the closure induced by the Galois
connection) holds, we consider a few examples where compositionality
fails:

\begin{example}\label{expl:non-compositional}
  Consider the powerset functor $\pfun$ with the predicate lifting
  $\sup$ ($\Lambda^\pfun = \{\sup\}$), and the discrete distribution
  functor $\dfun$ with the predicate lifting $\expect{}$ that takes
  expected values ($\Lambda^\dfun = \{\expect{}\}$).  With just these
  predicate liftings, compositionality fails for all four combinations
  $\pfun\pfun$, $\pfun\dfun$, $\dfun\pfun$, $\dfun\dfun$.  We show
  this for the case of $\pfun\dfun$ and discuss the others in
  \full{the appendix
    (\Cref{expl:non-compositional-full})}\short{\cite{dgkknrw:behavioural-metrics-compositionality-arxiv}}.
  Let $X$ be the two-element set $\{x,y\}$, equipped with the discrete
  metric $d$ (that is, $d(x,y) = d(y,x) = 1$), so that in particular
  all maps $g\colon X\to[0,1]$ are non-expansive.  We also consider
  the non-expansive function
  $f_\dfun\colon\nonexp{(\dfun
    X,\kant{\{\expect{}\}}d)}{([0,1],d_\unitQ)}$ given by
  \begin{math}
    f_\dfun(p\cdot x + (1-p)\cdot y) = \min(p,1-p).
  \end{math}
  Put $U = \{1\cdot x, 1\cdot y\}$ and $V = {\{1\cdot x,\nicefrac{1}{2}\cdot x+\nicefrac{1}{2}\cdot y,1\cdot y\}}$.
  Then
  \begin{math}
    \sup f_\dfun[U] = \max(0,0) = 0
  \end{math}
  and
  \begin{math}
    \sup f_\dfun[V] = \max(0,\nicefrac{1}{2},0) = \nicefrac{1}{2},
  \end{math}
  so that $\kant{\{\sup\}}(\kant{\{\expect{}\}}d)(U,V) \ge \nicefrac{1}{2}$.
  For every $g\colon X\to[0,1]$ one finds that
  \begin{align*}
    (\sup * \expect{})(g)(U) &= \max(g(x),g(y)) 
    = \textstyle\max(g(x),\nicefrac{g(x)+g(y)}{2},g(y)) = (\sup * \expect{})(g)(V),
  \end{align*}
  implying that $K_{\{\sup*\expect{}\}}d(U,V) = 0$.
\end{example}

\begin{toappendix}
  
\begin{example}\label{expl:non-compositional-full}
  Consider the powerset functor $\pfun$ with the predicate lifting
  $\sup$, and the discrete distribution functor $\dfun$ with the
  predicate lifting $\expect{}$ that takes expected values.  With just
  these predicate liftings, compositionality fails for all four
  combinations $\pfun\pfun$, $\pfun\dfun$, $\dfun\pfun$, $\dfun\dfun$.
  In all cases, let $X$ be the two-element set $\{x,y\}$, equipped
  with the discrete metric $d$ (that is, $d(x,y) = d(y,x) = 1$), so
  that in particular all maps $g\colon X\to[0,1]$ are non-expansive.
  We also consider the following two non-expansive functions on
  $\mathcal{P}X$ respectively $\mathcal{D}X$:
  \begin{align*}
    f_\pfun\colon\nonexp{(\pfun X,\kant{\{\sup\}}d)}{([0,1],d_\unitQ)} &&&
    f_\dfun\colon\nonexp{(\dfun X,\kant{\{\expect{}\}}d)}{([0,1],d_\unitQ)} \\
    f_\pfun(A) = \begin{cases}
      1, & \text{if $A=\{x,y\}$} \\
      0, & \text{otherwise}\\
    \end{cases}
    &&&
    f_\dfun(p\cdot x + (1-p)\cdot y) = \min(p,1-p)
  \end{align*}
  To see that $f_\mathcal{P}$ is indeed non-expansive, note that for
  $B\subseteq A$ we have $f_\mathcal{P}(B) \ominus f_\mathcal{P}(A) =
  0 = \kant{\{\sup\}}d(A,B)$ by monotonocity of $f_\mathcal{P}$ and
  $\sup$, and for $B\nsubseteq A$ we have $\kant{\{\sup\}}d(A,B) =
  1$ as we may pick $g$ to be the characteristic function of some $z\in B\setminus A$.
  In $\mathcal{D} X$, the distance of $\mu = p\cdot x + (1-p)\cdot y$ and $\nu = q\cdot x + (1-q)\cdot y$ can be computed as
  \begin{equation*}
    \kant{\{\expect{}\}}d(\mu,\nu)
    = \sup \{(q\cdot a + (1-q)\cdot b) - (p\cdot a + (1-p)\cdot b) \mid a,b\in[0,1] \}
    = |q-p|,
  \end{equation*}
  so that $f_\mathcal{D}$ is easily seen to be non-expansive as well.
  \begin{enumerate}
    \item\label{expl:pow-pow-fail} Put $U=\{\{x\},\{y\}\}$ and $V=\{\{x\},\{y\},\{x,y\}\}$.
      Then
      \begin{equation*}
        \sup f_\pfun[U] = 0 \quad\text{and}\quad
        \sup f_\pfun[V] = 1,
      \end{equation*}
      so that $\kant{\{\sup\}}(\kant{\{\sup\}} d)(U,V) = 1$.
      For every $g\colon X\to[0,1]$ one finds that
      \begin{equation*}
        (\sup * \sup)(g)(U) = \max(g(x),g(y)) = (\sup * \sup)(g)(V),
      \end{equation*}
      implying that $\kant{\{\sup * \sup\}}d(U,V) = 0$.
    \item\label{expl:pow-dist-fail} Put $U = \{1\cdot x, 1\cdot y\}$ and $V = \{1\cdot x,\nicefrac{1}{2}\cdot x+\nicefrac{1}{2}\cdot y,1\cdot y\}$.
      Then
      \begin{equation*}\textstyle
        \sup f_\dfun[U] = \max(0,0) = 0\quad\text{and}\quad
        \sup f_\dfun[V] = \max(0,\nicefrac{1}{2},0) = \nicefrac{1}{2},
      \end{equation*}
      so that $\kant{\{\sup\}}(\kant{\{\expect{}\}}d)(U,V) \ge \nicefrac{1}{2}$.
      For every $g\colon X\to[0,1]$ one finds that
      \begin{align*}
        (\sup * \expect{})(g)(U) &= \max(g(x),g(y)) 
        = \textstyle\max(g(x),\nicefrac{g(x)+g(y)}{2},g(y)) = (\sup * \expect{})(g)(V),
      \end{align*}
      implying that $\kant{\{\sup * \expect{}\}}(d)(U,V) = 0$.
    \item\label{expl:dist-pow-fail} Put $\mu = \nicefrac{1}{2}\cdot\{x\}+\nicefrac{1}{2}\cdot\{y\}$ and $\nu = 1\cdot\{x,y\}$.
      Then
      \begin{equation*}\textstyle
        \expect{\mu}(f_\pfun) = \nicefrac{1}{2}\cdot 0+\nicefrac{1}{2}\cdot 0 = 0
        \quad\text{and}\quad
        \expect{\nu}(f_\pfun) = 1\cdot 1 = 1,
      \end{equation*}
      so that $\kant{\{\expect{}\}}(\kant{\{\sup\}} d)(U,V) = 1$.
      For every $g\colon X\to[0,1]$ one finds that
      \begin{equation*}\textstyle
        (\expect{} * \sup)(g)(\mu) = \nicefrac{g(x)+g(y)}{2} \quad\text{and}\quad
        (\expect{} * \sup)(g)(\nu) = \max(g(x),g(y)),
      \end{equation*}
      implying that $\kant{\{\expect{} * \sup\}}(d)(\mu,\nu) \le \nicefrac{1}{2}$.
    \item\label{expl:dist-dist-fail} Put $\mu = \nicefrac{1}{2}\cdot(1\cdot x) + \nicefrac{1}{2}\cdot(1\cdot y)$ and $\nu = 1\cdot(\nicefrac{1}{2}\cdot x+\nicefrac{1}{2}\cdot y)$.
      Then
      \begin{equation*}\textstyle
        \expect{\mu}(f_\dfun) = 0 \quad\text{and}\quad
        \expect{\nu}(f_\dfun) = \nicefrac{1}{2},
      \end{equation*}
      so that $\kant{\{\expect{}\}}(\kant{\{\expect{}\}} d)(U,V) \ge \nicefrac{1}{2}$.
      For every $g\colon X\to[0,1]$ one finds that
      \begin{equation*}\textstyle
        (\expect{} * \expect{})(g)(\mu) = \nicefrac{g(x)+g(y)}{2} = (\expect{} * \expect{})(g)(\nu),
      \end{equation*}
      implying that $\kant{\{\expect{} * \expect{}\}}(d)(\mu,\nu) = 0$.
  \end{enumerate}
\end{example}

\end{toappendix}

\subsection{Finite Coproduct Polynomial Functors}
\label{sec:polynomial-functors}

\todo{Polynomial functor phrasing more precise (finite coproducts) +
  explanation of finite coproducts. B: Wojtek did this, I changed a
  tiny bit}

We now assume that the first functor ($F$ with the lifting
$\overline{F} = \kant{\Lambda_F}$) is in fact a polynomial functor
(with finite coproducts) and we show that in this case
compositionality holds automatically for certain sets of predefined
evaluation maps. This will later allow us to use compositionality to
define up-to techniques for large classes of coalgebras that are based
on such functors.

Consider the set of polynomial functors (with finite coproducts) given
by
\begin{equation*}
  F \Coloneqq \constfun{B} \mid \Id \mid \sprod_{i\in I} F_i \mid F_1 + F_2
\end{equation*}
where $\constfun{B}$ is the constant functor mapping to some set $B$ and $\Id$ is the identity functor.
We support products over arbitrary index sets, but we restrict to
finitary coproducts for simplicity.

For such polynomial functors we can obtain compositionality in a
structured manner, by constructing suitable sets of predicate liftings
alongside with the functors themselves.  We recursively define a set
$\Lambda^F$ of evaluation functions for each polynomial functor $F$ as
follows:
\begin{description}
\item[constant functors:] $F = \constfun{B}$. Here we choose
  $\Lambda^F$ to be any set of maps of type $B \to \VV$.
  (For instance, when $B=\VV$ we can put
  $\Lambda^F = \{\id_\VV\}$.)
  \item[identity functor:] $F = \Id$. We put $\Lambda^F = \{ \id_\VV \}$.
  \item[product functors:] $F = \sprod_{i\in I} F_i$. Put
    $\Lambda^F = \{\ev_i\circ\pi'_i \mid i\in I,
    \ev_i\in\Lambda^{F_i}\}$ where $\pi'_i\colon \sprod_{i\in
      I}F_i\VV\to F_i\VV$ are the projections.
  \item[coproduct functors:] $F = F_1 + F_2$. We put
    \begin{math}
      \Lambda^F =
      \{[\ev_1,\top] \mid \ev_1\in\Lambda^{F_1}\}\cup
      \{[\bot,\ev_2] \mid \ev_2\in\Lambda^{F_2}\}\cup
      \{[\bot,\top]\},
    \end{math}
    where $\top$ and $\bot$ denote constant maps into $\VV$.
  \end{description}

\begin{remark}\label{rem:evalfun-associative}
  We note that the construction for coproduct functors is associative,
  that is, for functors $F_1$, $F_2$ and $F_3$ the sets $\Lambda^{(F_1
    + F_2) + F_3}$ and $\Lambda^{F_1 + (F_2 + F_3)}$ coincide up to
  isomorphism.

  Note also that the exponentiation $F^AX = (FX)^A$ is special case of
  the product where $I=A$ and $F_i=F$ for all $i\in I$.
\end{remark}
\begin{remark}
  Throughout the paper, we restrict our attention to finite coproducts
  for the sake of simplicity, but we would like to note that our
  construction could be generalized to infinite sets. In general,
  since our lifting for the coproduct will be based on prioritization, we
  need to compare the sets in order of preference, i.e. have the
  additional structure of a well-order. This immediately works for
  countable sets.
\end{remark}


The choice of evaluation maps above induces the following
liftings, leading to the natural expected distances in the directed
case. 
  
\begin{propositionrep}
  Given a polynomial functor $F$ and a set of evaluation maps
  $\Lambda^F$ as defined above, the corresponding lifting
  $\overline{F} = \kant{\Lambda^F}$ is defined as follows on objects
  of $\Vgraph$: $\overline{F}(d_X) = d^F_X\colon FX\times FX\to\VV$
  where
  \begin{description}
  \item[constant functors:] $d^F_X\colon B\times B\to \VV$,
    $d^F_X(b,c) = \bigsqcap_{\ev\in\Lambda^F}
    d_\VV(\ev(b),\ev(c))$.
  \item[identity functor:] $d^F_X \colon X\times X\to \VV$ with
    $d^F_X = \alpha_X(\gamma_X(d))$.
  \item[product functors:]
    $d^F_X\colon \sprod_{i\in I} F_iX\times \sprod_{i\in I} F_iX\to
    \VV$,
    $d^F_X(s,t) = \bigsqcap_{i\in I} d_X^{F_i}(\pi_i(s),\pi_i(t))$
    where $\pi_i \colon \sprod_{i\in I} F_iX\to F_iX$ are the
    projections.
  \item[coproduct functors:]
    $d^F_X\colon (F_1X+F_2X)\times (F_1X+F_2X)\to \VV$, where
    
    $d^F_X(s,t) =
    \begin{cases}
      d^{F_i}(s,t) & \mbox{if $s,t\in F_iX$ for $i\in\{1,2\}$} \\
      \top & \mbox{if $s\in F_1X$, $t\in F_2X$} \\
      \bot & \mbox{if $s\in F_2X$, $t\in F_1X$}
    \end{cases}$
  \end{description}
\end{propositionrep}

\begin{proof}
  \begin{description}
  \item[constant functors:] let $b,c\in B$. We have, since $Ff$ is
    the identity:
    \begin{eqnarray*}
      d^F_X(b,c) & = & \bigsqcap_{\ev\in\Lambda^{F}}
      \bigsqcap_{f\in\gamma(d_X)} d_\VV(\ev(Ff(b)),\ev(Ff(c))) \\
      & = & \bigsqcap_{\ev\in\Lambda^{F}}
      d_\VV(\ev(b),\ev(c))
    \end{eqnarray*}
  \item[identity functor:] let $x,y\in X$. We have, since $Ff = f$ and
    $\ev$ is the identity:
    \begin{eqnarray*}
      d^F_X(x,y) & = & \bigsqcap_{\ev\in\Lambda^{F}}
      \bigsqcap_{f\in\gamma(d_X)} d_\VV(\ev(Ff(x)),\ev(Ff(y))) \\
      & = & \bigsqcap_{f\in\gamma(d_X)} d_\VV(f(x)),f(y))) =
      \alpha_X(\gamma_X(d_X))(x,y)
    \end{eqnarray*}
  \item[product functors:] Let $s,t\in \sprod_{i\in I} F_iX$. Let
    $\pi_i\colon \sprod_{i\in I} F_iX\to F_iX$,
    $\pi'_i\colon \sprod_{i\in I} F_i\VV\to F_i\VV$. Then
    \begin{eqnarray*}
      d^F_X(s,t) & = & \bigsqcap_{\ev\in\Lambda^F}
      \bigsqcap_{f\in\gamma(d_X)} d_\VV(\ev(Ff(s)),\ev(Ff(t))) \\
      & = & \bigsqcap_{i\in I} \bigsqcap_{\ev_i\in\Lambda^{F_i}}
      \bigsqcap_{f\in\gamma(d_X)}
      d_\VV(\ev_i(\pi'_i(Ff(s))),\ev_i(\pi'_i(Ff(t)))) \\
      & = & \bigsqcap_{i\in I} \bigsqcap_{\ev_i\in\Lambda^{F_i}}
      \bigsqcap_{f\in\gamma(d_X)}
      d_\VV(\ev_i(F_if(\pi_i(s))),\ev_i(F_if(\pi_i(t)))) \\
      & = & \bigsqcap_{i\in I} d_X^{F_i}(\pi_i(s),\pi_i(t )) \\
    \end{eqnarray*}
    using the fact that $\pi'_i\circ Ff = F_if\circ \pi_i$.
  \item[coproduct functors:] We use the properties of $d_\VV$ shown in
    \Cref{lem:residuation}. Whenever $s,t\in F_1X$, we have that
    \begin{eqnarray*}
      d^F_X(s,t) & = & \bigsqcap_{\ev_1\in\Lambda^{F_1}}
      \bigsqcap_{f\in\gamma(d_X)} d_\VV(\ev_1(Ff(s)),\ev_1(Ff(t)))
      \sqcap d_\VV(\bot,\bot) \\
      & = & d^{F_1}(s,t) \sqcap \top = d^{F_1}(s,t)
    \end{eqnarray*}
    If instead $s,t\in F_2X$, we obtain
    \begin{eqnarray*}
      d^F_X(s,t) & = & \bigsqcap_{\ev_2\in\Lambda^{F_2}}
      \bigsqcap_{f\in\gamma(d_X)} d_\VV(\ev_2(Ff(s)),\ev_2(Ff(t)))
      \sqcap d_\VV(\top,\top) \\
      & = & d^{F_2}(s,t) \sqcap \top = d^{F_2}(s,t)
    \end{eqnarray*}
    Now assume that $s\in F_1X$, $t\in F_2X$. Then:
    \begin{eqnarray*}
      d^F_X(s,t) & = & \bigsqcap_{\ev_1\in\Lambda^{F_1}}
      \bigsqcap_{f\in\gamma(d_X)} d_\VV(\ev_1(Ff(s)),\top)
      \mathop{\sqcap} \\
      && \qquad 
      \bigsqcap_{\ev_2\in\Lambda^{F_2}} \bigsqcap_{f\in\gamma(d_X)}
      d_\VV(\bot,\ev_1(Ff(t)))
      \sqcap d_\VV(\bot,\top) \\
      & = & \top\sqcap \top\sqcap \top = \top
    \end{eqnarray*}
    Now assume that $s\in F_2X$, $t\in F_1X$. Then, since
    $d_\VV(\bot,\top) = \bot$:
    \begin{eqnarray*}
      d^F_X(s,t) & = & \bigsqcap_{\ev_1\in\Lambda^{F_1}}
      \bigsqcap_{f\in\gamma(d_X)} d_\VV(\top,\ev_1(Ff(t)))
      \mathop{\sqcap} \\
      && \qquad 
      \bigsqcap_{\ev_2\in\Lambda^{F_2}} \bigsqcap_{f\in\gamma(d_X)}
      d_\VV(\ev_2(Ff(s)),\bot)) \sqcap
      d_\VV(\top,\bot) = \bot
    \end{eqnarray*}
  \end{description}
\end{proof}

Under this choice of evaluation functions we can show the following,
which implies compositionality (cf.~\Cref{sec:compositionality}):
\begin{propositionrep}
  \label{prop:comp-polynomial}
  For every polynomial functor $F$ and the corresponding set
  $\Lambda^F$ of evaluation maps (as above) we have
  $\Lambda^F\circ \gamma\circ\alpha \subseteq \gamma{F}\circ\alpha{F}
  \circ \Lambda^F$ on non-empty sets of predicates.
\end{propositionrep}
\begin{proof}
  Given $F$ and $\Lambda^F$, we need to show that for every set $X$, every $S\subseteq\VV^X$, every non-expansive $f\colon\nonexp{(X,\alpha_X(S))}{(\VV,d_\VV)}$, every $\ev\in\Lambda^F$ and every $t_1,t_2\in FX$ we have
  \begin{equation}
    \alpha_{FX}(\Lambda^F_X(S))(t_1,t_2) \vle d_\VV(\ev\circ Ff(t_1), \ev\circ Ff(t_2)).
    \label{eqn:compatible}
  \end{equation}
  We proceed by induction over $F$.
  \begin{itemize}
    \item For $F=\constfun{B}$, the set $\Lambda^F$ consists of some maps of type $B\to\VV$.
      Then for \eqref{eqn:compatible} to hold we need that for every $b_1,b_2\in B$ and every $\ev\in\Lambda^F$ we have
      \begin{equation*}
        \vMeet_{\ev'\in\Lambda^F} d_\VV(\ev'(b_1), \ev'(b_2)) \vle d_\VV(\ev(b_1),\ev(b_2)).
      \end{equation*}
      This holds automatically, regardless of the specific choice of $\Lambda^F$, as the term on the right is one of the terms in the meet on the left.
    \item For $F=\Id$, the equation \eqref{eqn:compatible} reduces to the statement that $f$ is non-expansive wrt.~$\alpha_X(S)$, which is true by assumption.
    \item For $F=\sprod_{i\in I} F_i$, let $s = (s_i)_{i\in I}, t = (t_i)_{i\in I} \in FX$ and fix some $j\in I$.
      Then we have for each $\ev\in\Lambda^{F_j}$ and each $g\colon X\to\VV$ that $\ev\circ\pi_j\circ Fg(s) = \ev\circ F_j g(s_j)$ and $\ev\circ\pi_j\circ Fg(t) = \ev\circ F_j g(t_j)$, and thus
      \begin{align*}
        &\alpha_X(\Lambda^F_X(S))(s,t) \\
        &= \vMeet_{i\in I}\vMeet_{\ev\in\Lambda^{F_i}}\vMeet_{g\in S} d_\VV(\ev\circ\pi_i\circ Fg(s), \ev\circ\pi_i\circ Fg(t)) \\
        &\vle \vMeet_{\ev\in\Lambda^{F_j}}\vMeet_{g\in S} d_\VV(\ev\circ\pi_j\circ Fg(s), \ev\circ\pi_j\circ Fg(t)) \\
        &= \vMeet_{\ev\in\Lambda^{F_j}}\vMeet_{g\in S} d_\VV(\ev\circ F_j g(s_j), \ev\circ F_j g(t_j)) \\
        &= \alpha_X(\Lambda^{F_j}_X(S))(s_j,t_j) \\
        &\vle d_\VV(\ev\circ F_j f(s_j), \ev\circ F_j f(t_j)) &&\by{IH} \\
        &= d_\VV(\ev\circ\pi_j\circ F f(s), \ev\circ\pi_j\circ F f(t))
      \end{align*}
      for each $f\colon\nonexp{(X,\alpha_X(S))}{(\VV,d_\VV)}$ and each $\ev\in\Lambda^{F_j}$.
    \item For $F=F_1 + F_2$, let $s\in F_i X \subseteq F X$ and $t\in F_j X \subseteq F X$, and let $\ev\in\Lambda^F$.
      \begin{itemize}
        \item If $i=1$ and $\ev$ is of the form $[\bot,\top]$ or $[\bot,\ev_2]$, then the right hand side of \eqref{eqn:compatible} evaluates to $d_\VV(\bot,v)$ for some $v\in\VV$, which equals $\top$.
        \item If $j=2$ and $\ev$ is of the form $[\bot,\top]$ or $[\ev_1,\top]$, then the right hand side of \eqref{eqn:compatible} evaluates to $d_\VV(v,\top)$ for some $v\in\VV$, which equals $\top$.
        \item If $i=j=1$, and $\ev$ is of the form $[\ev_1,\top]$, then we show ~\eqref{eqn:compatible} as follows.
          \begin{align*}
            &\alpha_X(\Lambda^F_X(S))(s,t) \\
            &\vle \vMeet_{\ev_1\in\Lambda^{F_1}} \vMeet_{g\in S} d_\VV([\ev_1,\top]\circ Fg(s), [\ev_1,\top]\circ Fg(t)) \\
            &= \vMeet_{\ev_1\in\Lambda^{F_1}} \vMeet_{g\in S} d_\VV(\ev_1\circ F_1 g(s), \ev_1\circ F_1 g(t)) \\
            &= \alpha_X(\Lambda^{F_1}_X(S))(s,t) \\
            &\vle d_\VV(\ev_1\circ F_1 f(s), \ev_1\circ F_1 f(t)) && \by{IH} \\
            &= d_\VV([\ev_1,\top]\circ Ff(s), [\ev_1,\top]\circ Ff(t))
          \end{align*}
        \item If $i=j=2$, and $\ev$ is of the form $[\bot,\ev_2]$, the proof is very similar to the previous item.
        \item The only remaining case is that where $i=2$ and $j=1$.
          In that case, we can pick $[\bot,\top]\in\Lambda^F$ and some arbitrary $g\in S$ (here it is important that $S$ is non-empty) to show that the left hand side of \eqref{eqn:compatible} evaluates to~$\bot$:
          \begin{equation*}
            \alpha_X(\Lambda^F_X(S))(s,t)
            \vle d_\VV([\bot,\top]\circ Fg(s), [\bot,\top]\circ Fg(t)) = d_\VV(\top,\bot) = \bot. \qedhere
          \end{equation*}
      \end{itemize}
  \end{itemize}
\end{proof}

Using the arguments of \Cref{sec:compositionality}, we infer:
\begin{corollary}
  Let $F$ and $G$ be functors, and $\Lambda^F$ and $\Lambda^G$ be sets of predicate liftings for them.
  If $F$ and $\lambda^F$ are as in~\Cref{prop:comp-polynomial}, then $\kant{\Lambda^F}\circ \kant{\Lambda^G} = \kant{\Lambda^F*\Lambda^G}$.
\end{corollary}

\begin{example}
  \label{ex:running-1}
  We consider a running example specifying standard directed trace
  metrics for probabilistic automata as introduced in
  \cite{jss:trace-determinization-journal}. We take the polynomial
  functor $F = [0,1]\times \_^A$ (``machine functor''), monad
  $T = \mathcal{D}$ and quantale $\VV = \unitQ$. Furthermore we use
  expectation ($\expect{}$) as evaluation map for $T$ and as
  evaluation maps for the functor $F$ we take $\ev_*$ mapping to the
  first component and $\ev_a$ (for each $a\in A$) with
  $\ev_a(r,g) = g(a)$. These evaluation maps are of the type described
  in this section and hence we have compositionality.

  Note that this example is not directly realizable in the Wasserstein
  approach \cite{bkp:up-to-behavioural-metrics-fibrations-journal}:
  the issue with the Wasserstein lifting is that whenever no coupling
  of two elements exists, the distance is automatically the bottom
  element in the quantale. This can be seen for the functor $F$ where
  $(r_1,x), (r_2,x)\in FX$ have no coupling whenever $r_1\neq
  r_2$. Hence it is harder to parameterize and would not work here.

  Using a set of evaluation maps $\Lambda^F$ as opposed to a single
  evaluation map gives us additional flexibility.
\end{example}

\section{Application: Up-To Techniques}

We now adapt results from
\cite{bkp:up-to-behavioural-metrics-fibrations-journal} on up-to
techniques from Wasserstein to Kantorovich liftings. In particular, we
instantiate the fibrational approach to coinductive proof techniques
from \cite{bppr:general-coinduction-up-to} that allows to prove lower
bounds for greatest fixpoints, using post-fixpoints up-to as
witnesses. As shown in the running example and in
\Cref{sec:case-study} this can greatly help to reduce the size of such
witnesses, even allowing finitary witnesses which would be infinite
otherwise.

\subsection{Introduction to Up-To Techniques}
\label{sec:intro-up-to}

We first recall the notion of a bialgebra \cite{j:bialgebraic-review},
a coalgebra with a compatible algebra structure.

\begin{definition}
  \label{def:bialgebra}
  Consider two functors $F,T$ and a natural transformation
  $\zeta\colon TF\Rightarrow FT$. An \emph{$F$-$T$-bialgebra} for
  $\zeta$ is a tuple $(Y,a,c)$ such that $a\colon TY\to Y$ is a
  $T$-algebra and $c\colon Y\to FY$ is an $F$-coalgebra so that the
  diagram below commutes.
  \[
    \begin{tikzcd}
      TY \ar[r,"a"] \ar[d,"Tc"] & Y \ar[r,"c"] & FY \\
      TFY \ar[rr,"\zeta_Y"] & & FTY \ar[u,"Fa"]
    \end{tikzcd}
  \]
\end{definition}

In order to construct such bialgebras, distributive laws exchanging
functors and monads are helpful. 

\begin{definition}
  \label{def:em-law}
  A \emph{distributive law} or \emph{EM-law} of a monad $T:\C\to \C$ with unit
  $\eta:\mathrm{Id}\Rightarrow T$ and multiplication
  $\mu:TT\Rightarrow T$ over
  a functor $F:\C\to \C$ is a natural transformation
  $\zeta:TF\Rightarrow FT$ such that the following diagrams commute:
  \[\begin{tikzcd}
    FX && {T^2FX} & TFTX & {FT^2X} \\
    TFX & FTX & TFX && FTX \arrow["{\eta_{FX}}"', from=1-1, to=2-1]
    \arrow["{\zeta_X}"', from=2-1, to=2-2] \arrow["{F\eta_X}",
    from=1-1, to=2-2] \arrow["{\mu_{FX}}"', from=1-3, to=2-3]
    \arrow["{F\mu_X}", from=1-5, to=2-5] \arrow["{T\zeta_X}",
    from=1-3, to=1-4] \arrow["{\zeta_{TX}}", from=1-4, to=1-5]
    \arrow["{\zeta_X}"', from=2-3, to=2-5]
  \end{tikzcd}\]
\end{definition}

Whenever $T$ is a monad and $\zeta$ is an EM-law, then an
$F$-$T$-bialgebra can be obtained by determinizing a coalgebra
$c\colon X\to FTX$. More concretely, we obtain $c^\#\colon Y\to FY$
where $Y=TX$ and $c^\# = F\mu_X\circ \zeta_{TX}\circ Tc$. The algebra
map is $a=\mu_X\colon TY\to Y$.


\todo{Change to F-T-bialgebra. B: done}
We now assume a bialgebra $(Y,a,c)$ and Kantorovich liftings
$\overline{T} = \kant{\Lambda^T}$,
$\overline{F} = \kant{\Lambda^F}$ of $T,F$.  Based on this we can
define a \emph{behaviour function} $\textrm{beh}$ via\todo{Recall
  notation of $c^*$? B: done}
\[ \Vgraph_Y \stackrel{\overline{F}}{\longrightarrow} \Vgraph_{FY}
  \stackrel{c^*}{\longrightarrow} \Vgraph_{Y} \] Remember that $c^*$
denotes reindexing via $c$. The greatest fixpoint of $\textrm{beh}$
corresponds to a behavioural conformance (e.g., behavioural
equivalence or bisimulation metric).

\begin{example}
  \label{ex:running-2}
  We continue with \Cref{ex:running-1}.  We use the standard
  distributive law $\zeta\colon TF\Rightarrow FT$ given by the following
  components where $\expect{\mu} \pi_1 = \expect{}(\mathcal{D}\pi_1(\mu))$:
  \begin{eqnarray*}
    \zeta_X\colon \mathcal{D}([0,1]\times X^A) & \to & [0,1]\times
    \mathcal{D}X^A \\
    \zeta_X(\mu) & = & (\expect{\mu} \pi_1, a\mapsto \dfun (\mathsf{eval}_a\circ\pi_2)(\mu))
  \end{eqnarray*}
  where $\mathsf{eval}_a(f) = f(a)$.
  Given an Eilenberg-Moore coalgebra $c\colon X\to FTX$ (more
  concretely: $c\colon X\to [0,1]\times \mathcal{D}X^A$) and its
  determinization $c^\#\colon TX\to FTX$, the behavioural distance on
  $TX$ arises as the greatest fixpoint (in the quantale order) of the
  map $\textrm{beh} = (c^\#)^*\circ \bar{F}$ defined above.

  By unravelling the fixpoint equation one can see that it coincides
  with the directed trace metric on probability distributions that is
  defined as follows: for each state $x\in X$ let
  $\mathit{tr}_x\colon A^*\to [0,1]$ be a map that assigns to each
  word (trace) $w\in A^*$ the expected payoff for this word when read
  from $x$, where the payoff of a state $x'$ is $\pi_1(c(x'))$. Then
  \[ \nu \textrm{beh}(p,q) = \sup_{w\in A^*} \big( \sum\nolimits_{x\in X}
    \mathit{tr}_x(w)\cdot q(x) \ominus \sum\nolimits_{x\in X}
    \mathit{tr}_x(w)\cdot p(x)\big) \]
  If $p,q$ are Dirac distributions $\delta_x,\delta_y$, we have:
  $\nu \textrm{beh}(\delta_x,\delta_y) = \sup_{w\in A^*}
  (\mathit{tr}_y(w) \ominus \mathit{tr}_x(w))$.
\end{example}

%
%
One can typically avoid computing the full fixpoint $\nu\textrm{beh}$
when checking the behavioural distance of two states;
this is facilitated through the use of an
\emph{up-to function} $u$ defined via
\[ \Vgraph_Y \stackrel{\overline{T}}{\longrightarrow} \Vgraph_{TY}
  \stackrel{\Sigma_a}{\longrightarrow} \Vgraph_{Y} \] where
$\Sigma_f\colon \Vgraph_X \to \Vgraph_Y$ is defined as
$\Sigma_f(d)(y_1,y_2) = \bigsqcup_{f(x_i)=y_i} d(x_1,x_2)$ for
$f\colon X\to Y$ (direct image).

Both functions ($\mathrm{beh}$, $u$) are monotone functions on a
complete lattice. Hence we can use the Knaster-Tarski theorem
\cite{t:lattice-fixed-point} and the theory of up-to techniques
\cite{p:complete-lattices-up-to}. In particular, given a monotone
function $f\colon L\to L$ over a complete lattice $(L,\sqsubseteq)$,
we have the guarantee that $\ell\sqsubseteq f(\ell)$ for $\ell\in L$
guarantees $\ell\sqsubseteq \nu f$, i.e., a post-fixpoint of $f$ is always a
lower bound for the greatest fixpoint $\nu f$, an essential proof
rule in coinductive reasoning. Even more widely applicable are proof
rules based on up-to functions. An up-to function is a monotone
function $u\colon L\to L$ that is $f$-compatible (i.e.,
$u\circ f\sqsubseteq f\circ u$). Then we can infer that
$\ell\sqsubseteq f(u(\ell))$ (i.e., $\ell$ is a post-fixpoint up-to
$u$) implies $\ell\sqsubseteq \nu f$ ($\ell$ is a lower bound for the
greatest fixpoint). Typically $u$ is extensive
($\ell\sqsubseteq u(\ell)$) and hence it is ``easier'' to find a
post-fixpoint up-to rather than a post-fixpoint.

From \cite{bppr:general-coinduction-up-to} we obtain
the following result that ensures compatibility:
\todo{Reviewer B: Ref Proposition 26
  \cite{bkp:up-to-behavioural-metrics-fibrations-journal}? B: we have
  to cite a different paper, done.} 
\begin{proposition}[\cite{bppr:general-coinduction-up-to}]
  Whenever the EM-law $\zeta\colon TF\Rightarrow FT$ lifts to
  $\zeta\colon \overline{T}\,\overline{F}\Rightarrow
  \overline{F}\,\overline{T}$, we have that
  $u\circ \mathrm{beh} \vle \mathrm{beh}\circ u$ (for $u$,
  $\mathrm{beh}$ as defined above), i.e., $u$ is
  $\mathrm{beh}$-compatible.
\end{proposition}

Hence, we can deduce that every post-fixpoint up-to witnesses a lower
bound of the greatest fixpoint. More concretely:
$d_Y\vle \textrm{beh}(u(d_Y))$ implies $d_Y\vle \nu\,\textrm{beh}$
(where $d_Y\in\Vgraph_Y$) (coinduction up-to proof principle).

\subsection{Lifting Distributive Laws}
\label{sec:distributive-law-lifting}

To use the proof technique laid out in the previous section, we have to show that $\zeta$ lifts
accordingly. We start by defining distributive laws and lifting them
to $\Vgraph$.

Let $F$ be a polynomial functor (cf.~\Cref{sec:polynomial-functors})
and $(T,\mu,\eta)$ a monad over $\Set$. Following
\cite[Exercise~5.4.4]{j:introduction-coalgebra}, EM-laws
$\zeta:TF\Rightarrow FT$ \todo{Reviewer A: ``Is the EM-law terminology
  here standard?'' B: I would not do anything here} can then be
constructed inductively over the structure of $F$, i.e., for $F$ being
an identity, constant, product and coproduct functor. In the coproduct
case we extend \cite{j:introduction-coalgebra} by weakening the
requirement that $T$ preserves coproducts.

In the following, we inductively construct an EM-law
$\zeta:TF\Rightarrow FT$ and first lift it to
$\zeta:\overline{TF}\Rightarrow \overline{FT}$ (and then to
$\zeta:\overline{T}\,\overline{F}\Rightarrow
\overline{F}\,\overline{T}$). For the evaluation maps we assume that
$\Lambda^F$ is defined as in \Cref{sec:polynomial-functors} and that
$\Lambda^T = \{\ev_T\}$, where $\ev_T\colon T\VV\to \VV$ is a
$T$-algebra.
\begin{description}
\item[constant functors:] For $F=\constfun{B}$ we have that $TFX=TB$
  and $FTX=B$, and so define the EM-law as
  $\zeta:T\constfun{B}\Rightarrow \constfun{B}$, where the (unique)
  component $\zeta_X\colon TB\to B$ is an arbitrary $T$-algebra on
  $B$.  (From now on, we assume that evaluation maps
  $\ev\colon B\to \VV$ for constant functors are $T$-algebra
  homomorphisms between $\zeta\colon TB\to B$ and
  $\ev_T\colon T\VV\to \VV$.)
\item[identity functor:] For $F=\operatorname{Id}$, we let the EM-law
  be the identity map $\operatorname{id}:T\Rightarrow T$.
\item[product functors:] For 
  $F=\sprod_{i\in I} F_i$, assuming we have distributive laws
  $\zeta^i:TF_i\Rightarrow F_i T$, the EM-law is
  \[\langle \zeta^i \circ T\pi_i\rangle: T\sprod_{i\in
      I}F_i\Rightarrow\sprod_{i\in I} F_iT.\]
\item[coproduct functors:] For $F=F_1+F_2$, assume that we have
  distributive laws $\zeta^i:TF_i\Rightarrow F_i T$ and a natural
  transformation $g:T ((-) + (-)) \Rightarrow T + T$ between
  bifunctors.  The EM-law is given by
  \[ (\zeta^1 +\zeta^2)\circ g_{F_1,F_2}\colon T(F_1+F_2) \Rightarrow
    TF_1+TF_2 \Rightarrow F_1 T + F_2 T\]
\end{description}


    



\begin{definition}
  \label{def:compatibility-g}
  Let $T$ be a monad and let $g:T ((-) + (-)) \Rightarrow T + T$ be a
  natural transformation as above.
  We say that $g$ is \emph{compatible with the unit} $\eta$ of the monad if
  for all sets $Y_1, Y_2$ the left diagram below commutes. Analogously,
  $g$ is \emph{compatible with the multiplication} $\mu$ of the monad if the
  right diagram commutes for all sets $Y_1, Y_2$.
  \begin{equation*}
    \adjustbox{scale=0.85}{
    \begin{tikzcd}
      & {Y_1 + Y_2} \\
      {T(Y_1 + Y_2)} && {TY_1 + TY_2} \arrow["{g_{Y_1,Y_2}}",
      from=2-1, to=2-3] \arrow["{\eta_{Y_1 + Y_2}}"', from=1-2,
      to=2-1] \arrow["{\eta_{Y_1} + \eta_{Y_2}}", from=1-2, to=2-3]
    \end{tikzcd}
    }
    \quad
    \adjustbox{scale=0.85}{
    \begin{tikzcd}
      {TT(Y_1 + Y_2)} && {T(Y_1 + Y_2)} \\
      {T(TY_1 + TY_2)} & {TTY_1 + TTY_2} & {TY_1 + TY_2}
      \arrow["{Tg_{Y_1,Y_2}}"', from=1-1, to=2-1]
      \arrow["{g_{TY_1, TY_2}}"', from=2-1, to=2-2]
      \arrow["{\mu_{Y_1} + \mu_{Y_2}}"', from=2-2, to=2-3]
      \arrow["{g_{Y_1, Y_2}}", from=1-3, to=2-3]
      \arrow["{\mu_{Y_1 + Y_2}}", from=1-1, to=1-3]
    \end{tikzcd}
    }
  \end{equation*}
\end{definition}

\begin{propositionrep}
  Assume that the natural transformation $g$ is compatible with unit
  and multiplication of $T$. Then the transformation $\zeta$ as
  defined above is an EM-law of $T$ over $F$.
\end{propositionrep}
\begin{proof}
  We will proceed by induction over the shape of $F$. We wish to show
  that $\zeta\colon TF\Rightarrow FT$ is a natural
  transformations and that the two diagrams in
  \Cref{def:em-law} commute.

\begin{description}
\item[constant functors:]
If $F=\constfun{B}$, then naturality is clear and the distributive law diagrams become
\[\begin{tikzcd}
	B && {T^2B} & TB & B \\
	TB & B & TB && B,
	\arrow["{\eta_X}"', from=1-1, to=2-1]
	\arrow["\zeta"', from=2-1, to=2-2]
	\arrow["{\operatorname{Id}_B}", from=1-1, to=2-2]
	\arrow["{\mu_{B}}"', from=1-3, to=2-3]
	\arrow["{\mathrm{Id}_B}", from=1-5, to=2-5]
	\arrow["T\zeta", from=1-3, to=1-4]
	\arrow["\zeta", from=1-4, to=1-5]
	\arrow["\zeta"', from=2-3, to=2-5]
\end{tikzcd}\]
which are precisely the conditions for the map $\zeta:TB\to B$ being an algebra over $T$.

\item[identity functor:]
For $F=\mathrm{Id}$, it is also clear that $\mathrm{id}:T\Rightarrow T$ is a natural transformation and a distributive law.

\item[product functors:]
For $F=\sprod_{i\in I}F_i$, we note that the definition of the distributive law $\zeta$ in terms of the $\zeta^i$ amounts to the set of equalities $\zeta^i\circ T\pi_i = \pi_i\circ\zeta$ for each $i\in I$.
Therefore we can verify the triangle and pentagon diagrams via the following computations:
\begin{align*}
  F\eta
  &= \langle F_i\eta\circ\pi_i\rangle \\
  &= \langle\zeta^i\circ\eta_{F_i}\circ\pi_i\rangle &&\by{IH}\\
  &= \langle\zeta^i\circ T\pi_i\circ\eta_F\rangle &&\by{$\eta$ natural}\\
  &= \langle\pi_i\circ\zeta\circ\eta_F\rangle = \zeta\circ\eta_F \\
  F\mu\circ\zeta\circ T\zeta
  &= \langle F_i\mu\circ\pi_i\circ\zeta\circ T\zeta\rangle \\
  &= \langle F_i\mu\circ\zeta^i\circ T\pi_i\circ T\zeta\rangle \\
  &= \langle F_i\mu\circ\zeta^i\circ T\zeta^i\circ TT\pi_i\rangle \\
  &= \langle \zeta^i\circ\mu_{F_i}\circ TT\pi_i\rangle &&\by{IH}\\
  &= \langle \zeta^i\circ T\pi_i\circ\mu_F\rangle &&\by{$\mu$ natural}\\
  &= \langle \pi_i\circ\zeta\circ\mu_F\rangle =\zeta\circ\mu_F
\end{align*}

\item[coproduct functors:]
For $F=F_1 + F_2$, we first check the commutativity of the triangular diagram. In the diagram below, the left inner triangle commutes because of the assumption on $g$, while the right inner triangle commutes by induction hypothesis.
\[\begin{tikzcd}[column sep=3.15em]
	{} & {F_1X + F_2 X} \\
	{T(F_1X + F_2 X)} & {TF_1X + TF_2 X} & {F_1T X + F_2 TX}
	\arrow["{g_{F_1 X,F_2 X}}"', from=2-1, to=2-2]
	\arrow["{\eta_{F_1X + F_2X}}"', from=1-2, to=2-1]
	\arrow["{\eta_{F_1X} + \eta_{F_2X}}", from=1-2, to=2-2]
	\arrow["{\zeta^1_X + \zeta^2_X}"', from=2-2, to=2-3]
	\arrow["{F_1\eta_X + F_2\eta_X}", from=1-2, to=2-3]
\end{tikzcd}\]
Hence, the outer diagram commutes, establishing the required condition.
Moving on to the pentagonal diagram, we have the following:
\[\hspace{-5mm}\begin{tikzcd}[column sep=3.45em]
	{TT(F_1 X + F_2 X)} & {T(TF_1 X + TF_2 X)} & {T(F_1TX + F_2TX)} \\
	& {TTF_1X + TTF_2X} & {TF_1TX + TF_2TX} \\
	&& {F_1 TT X + F_2 TTX} \\
	{T(F_1 X + F_2 X)} & {TF_1X + TF_2X} & {F_1 T X + F_2 T X}
	\arrow["{g_{F_1 T X, F_2 T X}}", from=1-3, to=2-3]
	\arrow["{\zeta^1_{TX} + \zeta^2_{TX}}", from=2-3, to=3-3]
	\arrow["{F_1\mu_X + F_2\mu_X}", from=3-3, to=4-3]
	\arrow["{\zeta^1_X + \zeta^2_X}", from=4-2, to=4-3]
	\arrow["{T\zeta^1_X + T\zeta^2_X}", from=2-2, to=2-3]
	\arrow["{\mu_{F_1 X} + \mu_{F_2X}}", from=2-2, to=4-2]
	\arrow["{T(\zeta^1_X + \zeta^2_X)}", from=1-2, to=1-3]
	\arrow["{g_{TF_1X, TF_2X}}", from=1-2, to=2-2]
	\arrow["{Tg_{F_1X, F_2X}}", from=1-1, to=1-2]
	\arrow["{g_{F_1 X, F_2 X}}", from=4-1, to=4-2]
	\arrow["{\mu_{F_1 X + F_2 X}}", from=1-1, to=4-1]
\end{tikzcd}\]
The top right square commutes by naturality of $g$, while the bottom right diagram commutes by induction hypothesis.
Finally, the left diagram commutes because of the assumption we make on $g$.
The commutativity of the outer diagram yields the desired result. \qedhere
\end{description}
\end{proof}



In order to lift natural transformations (respectively distributive
laws), we will use the following result:

\begin{propositionrep}
  \label{prop:lifting-em-law}
  Let $F,G$ be functors on $\Set$ and let the sets of evaluation maps
  of $F$ and $G$ be denoted by $\Lambda^{F}$ and $\Lambda^{G}$. Let
  $\zeta: F\Rightarrow G$ be a natural transformation. If
  \begin{equation}\label{eqn:eval-nat}
    \Lambda^{G} \circ \zeta_\VV := \{\ev_G\circ\zeta_\VV\mid
    \ev_G\in\Lambda^G\} \subseteq \Lambda^{F},
  \end{equation}
  then $\zeta$ lifts to
  $\zeta:\overline{F}\Rightarrow\overline{G}$ in $\Vgraph$, where
  $\overline{F} = \kant{\Lambda^F}$,
  $\overline{G} = \kant{\Lambda^G}$.
\end{propositionrep}
 
\begin{proof}
  To prove this, we need to show that the transformation $\zeta$ is
  non-expansive in each component, i.e. that for every
  $(X,d)\in \Vgraph$,
  \[d^F \vle d^G\circ (\zeta_X\times \zeta_X).\]
  Unpacking the definition of the Kantorovich lifting, we proceed as follows for $t_1,t_2\in FX$:
  \begin{eqnarray*}
    d^F(t_1,t_2) & = & \vMeet \big\{ d_{\VV}((\ev_F\circ
    Ff)(t_1),(\ev_F\circ
    Ff)(t_2))\mid \ev_F\in \Lambda^F, \\
    && \qquad\qquad \, f:X\to \VV, d\vle d_{\VV}\circ (f\times f) \big\}\\
    & \vle & \vMeet \big\{d_{\VV}((\ev_G\circ \zeta_\VV\circ
    Ff)(t_1),(\ev_G\circ \zeta_\VV\circ Ff)(t_2))\mid \ev_G\in
    \Lambda^G, \\
    && \qquad\qquad \, f:X\to \VV, d\vle d_{\VV}\circ (f\times f) \big\}\\
    & = & \vMeet \big\{d_{\VV}((\ev_G\circ Gf\circ
    \zeta_X)(t_1),(\ev_G\circ Gf\circ \zeta_X)(t_2))\mid \ev_G\in
    \Lambda^G, \\
    && \qquad\qquad \, f:X\to \VV, d\vle d_{\VV}\circ (f\times f) \big\}\\
    & = & d^G(\zeta_X(t_1),\zeta_X(t_2)).
  \end{eqnarray*}
  By the assumption we know that the set of functions of the form
  $\ev_G\circ \zeta_{\VV}$ is a subset of the set of functions of
  the form $\ev_F$, and so the first infimum is smaller or equal than
  the second. And by the naturality of $\zeta$, we have that
  $\zeta_{\VV}\circ Ff = Gf\circ \zeta_X$.
\end{proof}

We can show that the inclusion~\eqref{eqn:eval-nat} (even equality)
holds under some conditions.
  
\begin{definition}
  Let $g:T ((-) + (-)) \Rightarrow T + T$ be a natural transformation as
  introduced above and let $\ev_T\colon T\VV\to \VV$ be the evaluation
  map of the monad. We say that $g$ is \emph{well-behaved} wrt. $\ev_T$ if the
  following diagrams commute for $f_i\colon X_i\to \VV$, where
  $\bot,\top$ are constant maps of the appropriate type.
  \[
    \adjustbox{scale=0.9}{
      \begin{tikzcd}
	{T(X_1 + X_2)} && T\VV \\
	{TX_1 + TX_2} && \VV
	\arrow["{g_{X_1, X_2}}", from=1-1, to=2-1]
	\arrow["{ev_T}", from=1-3, to=2-3]
	\arrow["{T[f_1, \top_{X_2}]}", from=1-1, to=1-3]
	\arrow["{[\ev_T \circ Tf_1, \top_{TX_2}]}", from=2-1, to=2-3]
      \end{tikzcd}
    }
    \quad
    \adjustbox{scale=0.9}{
      \begin{tikzcd}
	{T(X_1 + X_2)} && T\VV \\
	{TX_1 + TX_2} && \VV
	\arrow["{g_{X_1, X_2}}", from=1-1, to=2-1]
	\arrow["{ev_T}", from=1-3, to=2-3]
	\arrow["{T[\bot_{X_1}, f_2]}", from=1-1, to=1-3]
	\arrow["{[\bot_{TX_1}, \ev_T \circ Tf_2]}", from=2-1, to=2-3]
      \end{tikzcd}
    }
    \quad
    \adjustbox{scale=0.9}{
      \begin{tikzcd}
	{T(X_1 + X_2)} && T\VV \\
	{TX_1 + TX_2} && \VV
	\arrow["{g_{X_1, X_2}}", from=1-1, to=2-1]
	\arrow["{ev_T}", from=1-3, to=2-3]
	\arrow["{T[\bot_{X_1}, \top_{X_2}]}", from=1-1, to=1-3]
	\arrow["{[\bot_{TX_1}, \top_{TX_2}]}", from=2-1, to=2-3]
      \end{tikzcd}}
  \]      
\end{definition}

\begin{lemmarep}
  \label{lem:eval-G-comp-zeta-equal-eval-F}
  Let $F$ be a polynomial functor and $T$ a monad with
  $\Lambda^T = \{\ev_T\}$.

  For distributive laws $\zeta$ as described above where the component
  $g$ is well-behaved wrt. $\ev_T$ and evaluation maps as defined in
  \Cref{sec:polynomial-functors}, we have that
  \[ (\Lambda^{F} * \Lambda^{T})\circ \zeta_\VV = \Lambda^{T} *
    \Lambda^{F}. \]
\end{lemmarep}
  
\begin{proof}
  We will prove this via the inductive definition of $F$. In
  particular we will show that for each $\ev\in \Lambda^T*\Lambda^F$
  there exists $\ev'\in\Lambda^F*\Lambda^T$ such that
  $\ev = \ev' \circ \zeta_V$. And for each
  $\ev'\in\Lambda^F*\Lambda^T$ there exists
  $\ev\in \Lambda^T*\Lambda^F$ such that $\ev = \ev' \circ \zeta_V$.
  
  \begin{itemize}
  \item
  For the case where $F=\constfun{B}$ we obtain:
  \begin{eqnarray*}
    \Lambda^F * \Lambda^T & = & \{\ev_F\circ F\ev_T \mid
    \ev_F\in\Lambda^F\} = \{\ev_F\circ \id_B \mid \ev_F\in\Lambda^F\}
    = \Lambda^F \\
    \Lambda^T * \Lambda^F & = & \{\ev_T\circ T\ev_F \mid
    \ev_F\in\Lambda^F\} = \{\ev_F\circ \zeta_\VV \mid
    \ev_F\in\Lambda^F\} 
  \end{eqnarray*}
  where the last equality holds since each map $\ev_F$ is an algebra
  map from $\zeta_\VV$ to $\ev_T$ by assumption.  In this case we
  clearly have the correspondence stated above.

  \item
  For $F=\operatorname{Id}$, recall that $\ev_F\in\Lambda^F$
  and $\zeta_\VV$ are both the identity. Hence:
  \begin{eqnarray*}
    \Lambda^F * \Lambda^T & = & \{\ev_F\circ F\ev_T \mid
    \ev_F\in\Lambda^F\} = \{\ev_T\} = \Lambda^T \\
    \Lambda^T * \Lambda^F & = & \{\ev_T\circ \id_{T\VV} \mid
    \ev_F\in\Lambda^F\} = \Lambda^T 
  \end{eqnarray*}
  Since both sets are the same and $\zeta_\VV$ is the identity, the
  correspondence clearly holds.
    
  \item
  Now, consider the case for the product functor
  $F=\sprod_{i\in I}F_i$. Then
  \begin{eqnarray*}
    \Lambda^F * \Lambda^T & = & \{\ev'_i\circ\pi_i \circ F\ev_T\mid
    i\in I,
    \ev'_i\in\Lambda^{F_i},\ev_T\in \Lambda^T\} \\
    \Lambda^T * \Lambda^F & = & \{\ev_T\circ T\ev_i \circ T\pi_i \mid
    i\in I, \ev_i\in\Lambda^{F_i},\ev_T\in \Lambda^T\}
  \end{eqnarray*}
  Here $\zeta_{\VV}:T\sprod_{i\in I}F_i\VV\to \sprod_{i\in I}F_iT\VV$ is
  given as the universal map
  $\langle \zeta^i_{\VV}\circ T\pi_i\rangle_{i\in I}$, i.e. the map
  such that $\pi'_i\circ \zeta_{\VV}=\zeta^i_{\VV}\circ T\pi_i$, where
  $\pi'_i\colon \sprod_{i\in I} F_iT\VV \to F_iT\VV$.

  Fix $i\in I$, $\ev_i \in \Lambda^{F_i}$. Then,
  $\ev_T\circ T\ev_i\in \Lambda^T*\Lambda^{F_i}$ and by the induction
  hypothesis there exists
  $\ev_i'\circ F_i\ev_T\in \Lambda^{F_i} * \Lambda^T$ with
  $\ev'_i\in\Lambda^{F_i}$ such that
  $\ev_T\circ T\ev_i = \ev_i'\circ F_i\ev_T\circ \zeta^i_\VV$.

  Now for $\ev_T \circ T\ev_i \circ T\pi_i \in \Lambda^T*\Lambda^F$ we
  pick $\ev'_i\circ\pi_i \circ F\ev_T \in \Lambda^F*\Lambda^T$ as
  corresponding map and we have:
  \begin{align*}
    \ev'_i\circ\pi_i \circ F\ev_T\circ \zeta_{\VV} 
    &= \ev'_i\circ\pi_i \circ F\ev_T\circ\langle \zeta^i_{\VV}\circ T\pi_i\rangle_i \\
    &= \ev'_i \circ \pi_i\circ\langle  F_i\ev_T\circ \zeta^i_{\VV}\circ T\pi_i\rangle\\
    &= \ev'_i \circ F_i\ev_T\circ \zeta^i_{\VV}\circ T\pi_i \\
    &= \ev_T \circ T\ev_i \circ T\pi_i 
  \end{align*}
  The proof is analogous for the other direction.
  
  \item
  Finally we address the case of the coproduct functor.  Then
  \begin{eqnarray*}
    \Lambda^F * \Lambda^T & = & \{[\ev_1\circ F\ev_T,\top] \mid
    \ev_1\in\Lambda^{F_1}\}\cup \{[\bot,\ev_2\circ F\ev_T] \mid
    \ev_2\in\Lambda^{F_2}\}\cup
    \{[\bot,\top]\} \\
    \Lambda^T * \Lambda^F & = & \{\ev_T\circ T[\ev_1,\top] \mid
    \ev_1\in\Lambda^{F_1}\}\cup \{\ev_T\circ T[\bot,\ev_2] \mid
    \ev_2\in\Lambda^{F_2}\}\cup
    \{\ev_T\circ T[\bot,\top]\}
  \end{eqnarray*}
  If we fix $\ev_i\in \Lambda^{F_i}$, then
  $\ev_T\circ T\ev_i\in \Lambda^T*\Lambda^{F_i}$ and by the induction
  hypothesis there exists
  $\ev_i'\circ F_i\ev_T\in \Lambda^{F_i} * \Lambda^T$ with
  $\ev'_i\in\Lambda^{F_i}$ such that
  $\ev_T\circ T\ev_i = \ev_i'\circ F_i\ev_T\circ \zeta^1_\VV$.

  We now consider three cases, according to the three subsets above.
  
  For $\ev_T\circ T[\ev_1,\top] \in \Lambda^T*\Lambda^F$ we pick
  $[\ev'_1\circ F\ev_T,\top] \in \Lambda^F*\Lambda^T$ as corresponding
  map and we have, using well-behavedness of $g$:
  \begin{align*}
    [\ev'_1\circ F\ev_T,\top]\circ \zeta_\VV & = [\ev'_1\circ F\ev_T,\top]
    \circ (\zeta^1_\VV+\zeta^2_\VV)\circ g_{F_1X,F_2X} \\
    & = [\ev'_1\circ F\ev_T\circ \zeta^1_\VV,\top\circ \zeta^2_\VV]
    \circ g_{F_1X,F_2X} \\
    & = [\ev'_1\circ F\ev_T\circ \zeta^1_\VV,\top]
    \circ g_{F_1X,F_2X} \\
    & = [\ev_T\circ T\ev_1,\top] \circ g_{F_1X,F_2X} \\
    & = \ev_T\circ T[\ev_1,\top]
  \end{align*}
  For $\ev_T\circ T[\bot,\ev_2] \in \Lambda^T*\Lambda^F$ we pick
  $[\bot,\ev'_2\circ F\ev_T] \in \Lambda^F*\Lambda^T$ as corresponding
  map and we have, using well-behavedness of $g$:
  \begin{align*}
    [\bot,\ev'_2\circ F\ev_T]\circ \zeta_\VV & = [\bot,\ev'_2\circ F\ev_T]
    \circ (\zeta^1_\VV+\zeta^2_\VV)\circ g_{F_1X,F_2X} \\
    & = [\bot\circ \zeta^1_\VV,\ev'_2\circ F\ev_T\circ \zeta^2_\VV]
    \circ g_{F_1X,F_2X} \\
    & = [\bot,\ev'_2\circ F\ev_T\circ \zeta^2_\VV]
    \circ g_{F_1X,F_2X} \\
    & = [\bot,\ev_T\circ T\ev_2] \circ g_{F_1X,F_2X} \\
    & = \ev_T\circ T[\bot,\ev_2]
  \end{align*}
  For $\ev_T\circ T[\bot,\top] \in \Lambda^T*\Lambda^F$ we pick
  $[\bot,\top] \in \Lambda^F*\Lambda^T$ as corresponding
  map and we have, using well-behavedness of $g$:
  \begin{align*}
    [\bot,\top]\circ \zeta_\VV & = [\bot,\top]
    \circ (\zeta^1_\VV+\zeta^2_\VV)\circ g_{F_1X,F_2X} \\
    & = [\bot\circ \zeta^1_\VV,\top \circ \zeta^2_\VV]
    \circ g_{F_1X,F_2X} \\
    & = [\bot,\top] \circ g_{F_1X,F_2X} \\
    & = \ev_T\circ T[\bot,\top]
  \end{align*}
  The other directions can be shown analogously. This concludes the
  proof. \qedhere
\end{itemize}
  
\end{proof}

Then, when we have a coalgebra of the form $Y\to FTY$ for $F$
polynomial and $T$ a monad as above, and we determinize it to get a
coalgebra $X\to FX$ for $X=TY$, we obtain a bialgebra with the algebra
structure given by the monad multiplication $\mu_Y:TX\to X$. The
EM-law obtained then also forms a distributive law for the
bialgebra. By \Cref{prop:lifting-em-law} and
\Cref{lem:eval-G-comp-zeta-equal-eval-F} we know that the distributive
law $\zeta$ lifts to $\Vgraph$, i.e.,
$\zeta\colon \overline{FT}\Rightarrow \overline{TF}$ where
$\overline{FT} = \kant{\Lambda^F*\Lambda^T}$ and
$\overline{TF} = \kant{\Lambda^T*\Lambda^F}$.

We now show that natural transformations $g$ as required above do
exist for the powerset and subdistribution monad for suitable
quantales. Note that they are ``asymmetric'' and prioritize one of the
two sets over the other.

\begin{propositionrep}
  Let $T=\mathcal{P}$ be the powerset monad with evaluation map
  $\ev_T = \sup$ for $\VV=\unitQ$. Then $g_{X_1,X_2}$ below is a
  natural transformation that is compatible with unit and
  multiplication of $T$ and is well-behaved.
  \[
    g_{X_1,X_2}\colon \mathcal{P}(X_1+X_2)  \to
    \mathcal{P}X_1+\mathcal{P}X_2 \qquad
    g_{X_1,X_2}(X') =
      \begin{cases}
        X'\cap X_1 & \mbox{if $X'\cap X_1\neq \emptyset$} \\
        X' & \mbox{otherwise}
      \end{cases}
  \]
\end{propositionrep}

\begin{proof}
  We have to prove the following properties of $g$:
  \begin{description}
  \item[naturality:] let $f_i\colon X_i\to Y_i$,
    $X'\subseteq X_1+X_1$. We distinguish the following cases: if
    $X'\cap X_1\neq \emptyset$, then
    \begin{align*}
      (\mathcal{P}(f_1)+\mathcal{P}(f_2))\circ g_{X_1,X_2}(X') & =
      (\mathcal{P}(f_1)+\mathcal{P}(f_2))(X'\cap X_1) \\
      & = \mathcal{P}(f_1)(X'\cap X_1) \\
      & = f_1[X'\cap X_1] \\
      & =
      g_{Y_1,Y_2}(f_1[X'\cap X_1]+f_2[X'\cap X_2]) \\
      & =
      g_{Y_1,Y_2}((f_1+f_2)[X']) \\
      & =
      g_{Y_1,Y_2}\circ \mathcal{P}(f_1+f_2)(X') 
    \end{align*}
    Note that the equality on the third line holds since $X'\cap X_1$
    is non-empty and the same is true for $f_1[X'\cap X_1]$.

    If $X'\cap X_1 = \emptyset$, hence $X'\subseteq X_2$, then
    \begin{align*}
      (\mathcal{P}(f_1)+\mathcal{P}(f_2))\circ g_{X_1,X_2}(X') & =
      (\mathcal{P}(f_1)+\mathcal{P}(f_2))(X') \\
      & = \mathcal{P}(f_2)(X') \\
      & = f_2[X'] \\
      & =
      g_{X_1,X_2}(f_2[X']) \\
      & =
      g_{X_1,X_2}((f_1+f_2)[X']) \\
      & =
      g_{X_1,X_2}\circ \mathcal{P}(f_1+f_2)(X') 
    \end{align*}
  \item[compatibility with $\eta$:]
    $g_{X_1,X_2}\circ \eta_{X_1+X_2} = \eta_{X_1}+\eta_{X_2}$

    Let $x\in X_1+X_2$. Then, since $g_{X_1,X_2}$ is the identity on
    singletons:
    \begin{align*}
      g_{X_1,X_2}\circ \eta_{X_1+X_2}(x) & = g_{X_1,X_2}(\{x\}) = \{x\}
    \end{align*}
  \item[compatibility with $\mu$:]
    $g_{X_1,X_2}\circ \mu_{X_1+X_2} = (\mu_{X_1}+\mu_{X_2}) \circ
    g_{TX_1,TX_2}\circ Tg_{X_1,X_2}$

    Let $\mathcal{X}\subseteq \mathcal{P}(X_1+X_2)$. We can split
    $\mathcal{X} = \mathcal{X}_1 + \mathcal{X}_2$, where
    $\mathcal{X}_1 = \{X'\in\mathcal{X}\mid X'\cap X_1\neq
    \emptyset\}$ and
    $\mathcal{X}_2 = \{X'\in\mathcal{X}\mid X'\cap X_1 = \emptyset\}
    \subseteq \mathcal{P}(X_2)$. Note that
    $\mathcal{X}_1\cap X_1 = \{X'\cap X_1\mid X'\in \mathcal{X}_1\}
    \subseteq \mathcal{P}(X_1)$. We consider two cases: if
    $\mathcal{X}_1\neq \emptyset$, we have that $\bigcup \mathcal{X}$
    contains at least one element from $X_1$ and hence:
    \begin{align*}
      g_{X_1,X_2}\circ \mu_{X_1+X_2}(\mathcal{X}) & =
      g_{X_1,X_2}(\bigcup \mathcal{X}) \\
      & = \left(\bigcup \mathcal{X}\right)\cap X_1 \\
      & = \bigcup (\mathcal{X}\cap X_1)  \\
      & = \bigcup (\mathcal{X}_1\cap X_1) \\
      & = (\mu_{X_1}+\mu_{X_2})(\mathcal{X}_1\cap X_1) \\
      & = (\mu_{X_1}+\mu_{X_2}) \circ
      g_{\mathcal{P}(X_1),\mathcal{P}(X_2)}((\mathcal{X}_1\cap X_1)
      + \mathcal{X}_2) \\
      & = (\mu_{X_1}+\mu_{X_2}) \circ
      g_{\mathcal{P}(X_1),\mathcal{P}(X_2)}\circ
      g_{X_1,X_2}[\mathcal{X}_1 + \mathcal{X}_2] \\
      & = (\mu_{X_1}+\mu_{X_2}) \circ
      g_{\mathcal{P}(X_1),\mathcal{P}(X_2)}\circ
      \mathcal{P}(g_{X_1,X_2})(\mathcal{X})
    \end{align*}
    If instead $\mathcal{X}_1 = \emptyset$, we have that
    $\bigcup \mathcal{X}$ contains no element from $X_1$ and
    $\mathcal{X} = \mathcal{X}_2$. Then:
    \begin{align*}
      g_{X_1,X_2}\circ \mu_{X_1+X_2}(\mathcal{X}) & = g_{X_1,X_2}\circ
      \mu_{X_1+X_2}(\mathcal{X}_2) \\
      & = g_{X_1,X_2}(\bigcup \mathcal{X}_2) \\
      & = \bigcup \mathcal{X}_2 \\
      & = (\mu_{X_1}+\mu_{X_2})(\mathcal{X}_2) \\
      & = (\mu_{X_1}+\mu_{X_2}) \circ
      g_{\mathcal{P}(X_1),\mathcal{P}(X_2)}(\mathcal{X}_2) \\
      & = (\mu_{X_1}+\mu_{X_2}) \circ
      g_{\mathcal{P}(X_1),\mathcal{P}(X_2)}\circ
      g_{X_1,X_2}[\mathcal{X}_2] \\
      & = (\mu_{X_1}+\mu_{X_2}) \circ
      g_{\mathcal{P}(X_1),\mathcal{P}(X_2)}\circ
      \mathcal{P}(g_{X_1,X_2})(\mathcal{X})
    \end{align*}
  \item[well-behavedness:] First, note that here $\top=0$ and $\bot = 1$.
    \begin{itemize}
    \item
      $[\sup\circ \mathcal{P}f_1,\top_{\mathcal{P}X_2}] \circ g_{X_1,X_2} = \sup\circ
      \mathcal{P}[f_1,\top_{X_2}]$:
      
      Let $X' \subseteq X_1 + X_2$. We distinguish three subcases.
      \begin{itemize}
      	\item If $X' \cap X_1 \neq \emptyset$, then
      	\begin{align*}
      		([\sup\circ \mathcal{P}f_1,\top_{\mathcal{P}X_2}] \circ g_{X_1,X_2})(X') &= (\sup \circ \mathcal{P}f_1)(X' \cap X_1) \\
      		&= \sup f_1[X' \cap X_1]\\
      		&= \sup (f_1[X' \cap X_1] \cup \top_{X_2}[X' \cap X_2])\\
      		&= \sup \left(\mathcal{P}[f_1, \top_{X_2}](X')\right)
      	\end{align*} 
      	\item If $X' = \emptyset$, then
     	\begin{align*}
     		([\sup\circ \mathcal{P}f_1,\top_{\mathcal{P}X_2}] \circ g_{X_1,X_2})(X') &= \top_{\mathcal{P} X_2}(\emptyset)\\
     				&= \top \\
     				&= \sup (\emptyset) \\
     				&= \sup ([\mathcal{P}f_1, \top_{X_2}](\emptyset)) \\
     				&= (\sup \circ [\mathcal{P}f_1, \top_{X_2}]) (X')
     	\end{align*}
     	\item $X' \cap X_1 = \emptyset$ and $X' \cap X_2 \neq
          \emptyset$, hence $X' \subseteq X_2$. We can assume that $X'\neq\emptyset$.
     	\begin{align*}
     		([\sup\circ \mathcal{P}f_1,\top_{\mathcal{P}X_2}] \circ g_{X_1,X_2})(X') &= \top_{\mathcal{P} X_2}(X' \cap X_2)\\
     				&= \top \\
     				&= \sup (\top_{X_2}(X' \cap X_2)) \\
     				&= \sup ([\mathcal{P}f_1, \top_{X_2}](X' \cap X_2)) \\
     				&= (\sup \circ [\mathcal{P}f_1, \top_{X_2}]) (X')
     	\end{align*}
      \end{itemize}
    \item
      $[\bot_{\mathcal{P}X_1},\sup\circ \mathcal{P}f_2] \circ g_{X_1,X_2} = \sup\circ
      \mathcal{P}[\bot_{X_1},f_2]$: 
      
      Let $X' \subseteq X_1 + X_2$. We distinguish two subcases.
      \begin{itemize}
      	\item If $X' \cap X_1 \neq \emptyset$, then:
      	\begin{align*}
      		([\bot_{\mathcal{P}X_1},\sup\circ \mathcal{P}f_2] \circ g_{X_1,X_2})(X') &= \bot_{\mathcal{P}X_1}(X' \cap X_1) \\
      		&= \bot \\
      		&= \sup (\bot_{X_1}[X' \cap X_1] \cup f_2 [X' \cap X_2]) \\
      		&= (\sup \circ \mathcal{P}[\bot_{X_1}, f_1])(X')
      	\end{align*}
      	\item If $X' \cap X_1 = \emptyset$, hence $X' \subseteq X_2$:
      	\begin{align*}
      		([\bot_{\mathcal{P}X_1},\sup\circ \mathcal{P}f_2] \circ g_{X_1,X_2})(X') &= (\sup \circ \mathcal{P}f_2)(X') \\
      		&= \sup (f_2[X']) \\
      		&= \sup (\bot_{X_1}[X' \cap X_1] \cup f_2[X' \cap X_2])\\
      		&= (\sup \circ \mathcal{P}[\bot_{X_1}, f_2])(X')
      	\end{align*}
      \end{itemize}
    \item
      $[\bot_{\mathcal{P}X_1},\top_{\mathcal{P}X_2}] \circ g_{X_1,X_2} = \sup\circ
      \mathcal{P}[\bot_{X_1},\top_{X_2}]$:
      Let $X' \subseteq X_1 + X_2$. We distinguish three subcases:
      \begin{itemize}
      	\item If $X' \cap X_1 \neq \emptyset$, then
      	\begin{align*}
      		([\bot_{\mathcal{P}X_1},\top_{\mathcal{P}X_2}] \circ g_{X_1,X_2})(X') &= \bot_{\mathcal{P} X_1} (X' \cap X_1) \\
      		&= \bot \\
      		&= \sup (\bot_{X_1}[X' \cap X_1]) \\
      		&= \sup (\bot_{X_1}[X' \cap X_1] \cup \top_{X_2}[X' \cap X_2]) \\
      		&= (\sup \circ \mathcal{P}[\bot_{X_1}, \top_{X_2}])(X')
      	\end{align*}
      	\item If $X' = \emptyset$, then
      	\begin{align*}
      		([\bot_{\mathcal{P}X_1},\top_{\mathcal{P}X_2}] \circ g_{X_1,X_2})(X') &= \top_{\mathcal{P} X_2} (\emptyset) \\
      		&= \top \\
      		&= \sup(\emptyset) \\
      		&= (\sup \circ \mathcal{P}[\bot_{X_{1}}, \top_{X_2}])(\emptyset) \\
      		&= (\sup \circ \mathcal{P}[\bot_{X_{1}}, \top_{X_2}])(X') \\
      	\end{align*}
      	\item $X' \cap X_1 = \emptyset$ and $X' \cap X_2 \neq \emptyset$, hence $X' \subseteq
          X_2$. We can assume that $X'\neq\emptyset$.
      	\begin{align*}
      		([\bot_{\mathcal{P}X_1},\top_{\mathcal{P}X_2}] \circ g_{X_1,X_2})(X') &= \top_{\mathcal{P} X_2}(X' \cap X_2) \\
      		&= \top \\
      		&= \sup(\top_{X_2}(X' \cap X_2)) \\
      		&= \sup(\bot_{X_1}[X' \cap X_1] \cup \top_{X_2}[X' \cap X_2]) \\
      		&= (\sup \circ \mathcal{P}[\bot_{X_1}, \top_{X_2}])(X') \qedhere  
      	\end{align*}
      \end{itemize}
    \end{itemize}
  \end{description}
\end{proof}

\begin{propositionrep}
  Let $T=\mathcal{S}$ be the subdistribution monad where
  $\mathcal{S}(X) = \{p\colon X\to [0,1]\mid \sum_{x\in X} p(x) \le
  1\}$. Assume that its evaluation map is $\ev_T=\expect{}$ for the
  quantale $\VV=\zeroinfQ$ (where we assume that
  $p\cdot \infty = \infty$ if $p>0$ and $0$ otherwise). Then $g_{X_1,X_2}$ below is a
  natural transformation that is compatible with unit and
  multiplication of $T$ and is well-behaved.
  \[
    g_{X_1,X_2}\colon \mathcal{S}(X_1+X_2) \to
    \mathcal{S}X_1+\mathcal{S}X_2 \qquad
    g_{X_1,X_2}(p) =
      \begin{cases}
        p|_{X_1} & \mbox{if $\supp(p)\cap X_1\neq \emptyset$} \\
        p|_{X_2} & \mbox{otherwise}
      \end{cases}
  \]
\end{propositionrep}

\begin{proof}
  We have to prove the following properties of $g$:
  \begin{description}
  \item[naturality:] let $f_i\colon X_i\to Y_i$,
    $p \in \mathcal{S}(X_1+X_1)$. We distinguish the following cases:
    if $\supp(p) \cap X_1 \neq \emptyset$, then for all $y \in Y_1$ we have that
     \begin{align*}
      \left((\mathcal{S}(f_1)+\mathcal{S}(f_2))\circ g_{X_1,X_2}(p)\right)(y) & =
      \left((\mathcal{S}(f_1)+\mathcal{S}(f_2))(p|_{X_1})\right)(y) \\
      & = \left(\mathcal{S}(f_1)(p|_{X_1})\right)(y) \\
      & = \sum_{y = f_1(x)} p|_{X_1}(x) \\
      & = \sum_{y = (f_1 + f_2)(x)} p|_{X_1}(x) \\
      &=\left(\mathcal{S}(f_1 + f_2)(p|_{X_1})\right)(y) \\
      &= \left(\mathcal{S}(f_1 + f_2)(p)\right)(y) \\ 
      &=(\mathcal{S}(f_1 + f_2)(p))|_{Y_1}(y) \\ 
      &= \left(g_{Y_1, Y_2} \circ \mathcal{S}(f_1 + f_2) (p)\right) (y)
    \end{align*}
    
    For the remaining case, we have that for all $y \in Y_2$
    \begin{align*}
    	\left((\mathcal{S}(f_1)+\mathcal{S}(f_2))\circ g_{X_1,X_2}(p)\right)(y) 
    	&= \mathcal{S}(f_2)(p|_{X_2})(y) \\
    	&= \sum_{y = f_2(x)}p|_{X_2}(x) \\
    	&= \sum_{y = (f_1 + f_2)(x)}p_{X_2}(x) \\
    	&= \left(\mathcal{S}(f_1 + f_2)(p|_{X_2})\right)(y) \\
    	&= \left(\mathcal{S}(f_1 + f_2)(p)\right)(y) \\	
    	&= (\mathcal{S}(f_1 + f_2)(p))|_{Y_2}(y) \\	
    	&= \left(g_{Y_1, Y_2} \circ \mathcal{S}(f_1 + f_2)(p)\right)(y)
    \end{align*}
  \item[compatibility with $\eta$:] $g_{X_1,X_2}\circ \eta_{X_1+X_2} = \eta_{X_1}+\eta_{X_2}$. Since $g_{X_1, X_2}$ is an identity when applied to Dirac distributions, we have that:
  	\[
		g_{X_1, X_2} \circ \eta_{X_1 + X_2}(x) = g_{X_1, X_2} (\delta_x) = \delta_x = \eta_{X_1}(x) 
	\]
  \item[compatibility with $\mu$:]
    $g_{X_1,X_2}\circ \mu_{X_1+X_2} = (\mu_{X_1}+\mu_{X_2}) \circ
    g_{\mathcal{S}X_1,\mathcal{S}X_2}\circ \mathcal{S}g_{X_1,X_2}$.
		Let $p \in \mathcal{S}\mathcal{S}(X_1 + X_2)$. 
		Note that we can partition $\supp(p) = \mathcal{X}_1 \cup \mathcal{X}_2$ in a way that $\mathcal{X}_1 = \{ \nu \in \supp(p) \mid \supp(\nu) \cap X_1 \neq \emptyset\}$ and $\mathcal{X}_2 = \{ \nu \in \supp(p) \mid \supp(\nu) \cap X_1 = \emptyset\}$. We consider two cases: if $\mathcal{X}_1 \neq \emptyset$, then $\supp(\mu_{X_1 + X_2}(p))\cap X_1 \neq \emptyset$ and hence for all $x \in X_1$, we have that:
		\begin{align*}
			\left((g_{X_1, X_2} \circ \mu_{X_1 + X_2})(p)\right)(x) &= (\mu_{X_1 + X_2}(p))|_{X_1}(x) \\
			&= \sum_{\nu \in \mathcal{S}(X_1 + X_2)} p(\nu) \nu|_{X_1}(x) \\
			&= \sum_{\nu \in \mathcal{X}_1} p(\nu)\nu|_{X_1}(x) \\
			&= \sum_{\nu \in \mathcal{X}_1} \left(\mathcal{S}(g_{X_1, X_2})(p)\right)(\nu)\nu(x) \\
			&=  \sum_{\nu \in \mathcal{S}(X_1 + X_2)} \left(\mathcal{S}g_{X_1,X_2}(p)|_{\mathcal{S}X_1}\right)(\nu) \nu(x)  \\
			&=  \left((\mu_{X_1}+\mu_{X_2}) \circ
    g_{\mathcal{S}X_1,\mathcal{S}X_2}\circ \mathcal{S}g_{X_1,X_2}(p)\right)(x)
		\end{align*}
		Now, consider the case when $\supp(\mu_{X_1 + X_2}(p)) \cap X_1 = \emptyset$. For all $x \in X_2$ we have that:
		\begin{align*}
			\left((g_{X_1, X_2} \circ \mu_{X_1 + X_2})(p)\right)(x) &= (\mu_{X_1 + X_2}(p))|_{X_2}(x) \\
			&= \sum_{\nu \in \mathcal{S}(X_1 + X_2)} p(\nu) \nu|_{X_2}(x) \\
			&= \sum_{\nu \in \mathcal{X}_2} p(\nu) \nu|_{X_2}(x) \\
			&= \sum_{\nu \in \mathcal{X}_2} p(\nu) \nu(x) \\
			&= \sum_{\nu \in \mathcal{X}_2} \left(\mathcal{S}(g_{X_1, X_2})(p)\right)(\nu)\nu(x) \\
			&= \sum_{\nu \in \mathcal{S}(X_1 + X_2)} \left(\mathcal{S}(g_{X_1, X_2})(p)|_{\mathcal{S} X_2}\right)(\nu)\nu(x) \\
			&= \sum_{\nu \in \mathcal{S}(X_1 + X_2)} \left(\mathcal{S}(g_{X_1, X_2})(p)|_{\mathcal{S} X_2}\right)(\nu)\nu(x) \\
			&=  \left((\mu_{X_1}+\mu_{X_2}) \circ
    g_{\mathcal{S}X_1,\mathcal{S}X_2}\circ \mathcal{S}g_{X_1,X_2}(p)\right)(x)
		\end{align*}
		
  \item[well-behavedness:] First, note that here $\top=0$ and $\bot = \infty$.
    \begin{itemize}
    \item
      $[\expect{}\circ {\mathcal{S}}f_1,\top_{{\mathcal{S}}X_2}] \circ g_{X_1,X_2} = \expect{}\circ
      {\mathcal{S}}[f_1,\top_{X_2}]$: Let $p \in S(X_1 + X_2)$. We first consider the case when $\supp(p) \cap X_1 \neq \emptyset$. We have that:
      \begin{align*}
      	[\expect{}\circ {\mathcal{S}}f_1,\top_{{\mathcal{S}}X_2}] \circ g_{X_1,X_2} (p) &= [\expect{}\circ {\mathcal{S}}f_1,\top_{{\mathcal{S}}X_2}] (p|_{X_1}) \\
      	&= \expect{}\circ {\mathcal{S}}f_1 (p|_{X_1}) \\
      	&= \sum_{y \in \VV} y \left(\sum_{y = f_1(x)} p(x)\right)\\
      	&= \left(\sum_{y \in \VV} y \left(\sum_{y = f_1(x)} p(x)\right)\right) + 0 \left(\sum_{0 = \top_{X_2}(x)} p(x)\right)\\
      	&= \sum_{y \in \VV} y \left(\sum_{y = [f_1,\top_{X_2}](x)} p(x)\right)\\
      	&= \expect{} \circ \mathcal{S} [f_1, \top_{X_2}]
      \end{align*}
     
     Now, consider the case when $\supp(p) = \emptyset$. We have that:
     \begin{align*}
     	[\expect{}\circ {\mathcal{S}}f_1,\top_{{\mathcal{S}}X_2}] \circ g_{X_1,X_2} (p) &= [\expect{}\circ {\mathcal{S}}f_1,\top_{{\mathcal{S}}X_2}] (p|_{X_2}) \\
     	&= \top \\
     	&= 0 \\
     	&= \expect{} \circ \mathcal{S}[f_1, \top_{X_2}] (p)
     \end{align*}
     Finally, consider the case when $\supp(p)\cap X_1 = \emptyset$, but $\supp(p) \neq \emptyset$. We have that:
     \begin{align*}
     	[\expect{}\circ {\mathcal{S}}f_1,\top_{{\mathcal{S}}X_2}] \circ g_{X_1,X_2} (p) &= [\expect{}\circ {\mathcal{S}}f_1,\top_{{\mathcal{S}}X_2}] (p|_{X_2}) \\
     	&= \top \\
     	&= 0 \\
     	&= \sum_{y \in \VV} y \left( \sum_{y = \top_{X_2}(x)} p(x)\right) \\
     	&= \expect{} \circ \mathcal{S}(\top_{X_2}) (p) \\
     	&= \expect{} \circ \mathcal{S}[f_1, \top_{X_2}] (p)
     \end{align*}
    \item
      $[\bot_{\mathcal{S}X_1},\expect{}\circ \mathcal{S}f_2] \circ g_{X_1,X_2} = \expect{}\circ
      \mathcal{S}[\bot_{X_1},f_2]$: Let $p \in S(X_1 + X_2)$. Firstly, consider the case when $\supp(p) \cap X_1 \neq \emptyset$. In such a case, we have
      \begin{align*}
      	[\bot_{\mathcal{S}X_1},\expect{}\circ \mathcal{S}f_2] \circ g_{X_1,X_2}(p) &= [\bot_{\mathcal{S}X_1},\expect{}\circ \mathcal{S}f_2] (p|_{X_1}) \\
      	&= \bot \\
      	&= \infty \\
      	&= \infty  + \left(\sum_{y \in \mathbb{R}} y \left( \sum_{y = [\bot_{X_1}, f_2](x)}p(x)\right)\right)\\
      	&= \left( \sum_{\infty = [\bot_{X_1}, f_2](x)}p(x)\right)  + \left(\sum_{y \in \mathbb{R}} y \left( \sum_{y = [\bot_{X_1}, f_2](x)}p(x)\right)\right)\\
      	&= \sum_{y \in [0, + \infty]} y \left( \sum_{y = [\bot_{X_1}, f_2](x)}p(x)\right)\\
      	&= \expect{}\circ \mathcal{S}[\bot_{X_1},f_2](p)
      \end{align*}
      If $\supp(p) = \emptyset$, then we have that
      \begin{align*}
      	[\bot_{\mathcal{S}X_1},\expect{}\circ \mathcal{S}f_2] \circ g_{X_1,X_2}(p) &= 	[\bot_{\mathcal{S}X_1},\expect{}\circ \mathcal{S}f_2](p|_{X_2}) \\
      	&= \expect{}\circ \mathcal{S}f_2(p|_{X_2}) \\
      	&= 0 \\
      	&= \expect{} \circ \mathcal{S}[\bot_{X_1}, f_2] (p)
      \end{align*}
      Finally, consider the case when $\supp(p) \cap X_1 = \emptyset$, but $\supp(p) \neq \emptyset$. Then, we have that
      \begin{align*}
      	[\bot_{\mathcal{S}X_1},\expect{}\circ \mathcal{S}f_2] \circ g_{X_1,X_2}(p) &= [\bot_{\mathcal{S}X_1},\expect{}\circ \mathcal{S}f_2]  (p|_{X_2}) \\
      	&= \expect{} \circ \mathcal{S}f_2 (p|_{X_2})\\
      	&= \sum_{y \in \VV} y \left(\sum_{y = f_2(x)} p|_{X_2}(x)\right)\\
      	&= \sum_{y \in \VV} y \left(\sum_{y = [\bot_{X_1}, f_2](x)} p(x)\right)\\
      	&= \expect{} \circ \mathcal{S}[\bot_{X_1}, f_2](p)
      \end{align*}
    \item
      $[\bot_{\mathcal{S} X_1},\top_{\mathcal{S} X_2}] \circ g_{X_1,X_2} = \expect{} \circ
      \mathcal{S}[\bot_{X_1},\top_{X_2}]$: Let $p \in \mathcal{S}(X_1 + X_2)$. First, consider the case when $\supp(p) \cap X_1 \neq \emptyset$. We have that
      \begin{align*}
      	[\bot_{\mathcal{S} X_1},\top_{\mathcal{S} X_2}] \circ g_{X_1,X_2}(p) &= [\bot_{\mathcal{S} X_1},\top_{\mathcal{S} X_2}]  (p|_{X_1}) \\
      	&= \bot_{\mathcal{S} X_1} (p|_{X_1}) \\
      	&= \infty \\
      	&= \infty  + 0\\
      	&= \left( \sum_{\infty = [\bot_{X_1}, \top_{X_2}](x)}p(x)\right)  + \left( \sum_{0 = [\bot_{X_1}, \top_{X_2}](x)}p(x)\right)\\
      	&= \sum_{y \in [0, + \infty]} y \left( \sum_{y = [\bot_{X_1}, \top_{X_2}](x)}p(x)\right)\\
      	&= \expect{}\circ \mathcal{S}[\bot_{X_1},\top_{X_2}](p)
      \end{align*}
      Then, consider the case when $\supp(p) = \emptyset$. 
      \begin{align*}
      [\bot_{\mathcal{S} X_1},\top_{\mathcal{S} X_2}] \circ g_{X_1,X_2}(p) &= [\bot_{\mathcal{S}X_1}, \top_{\mathcal{S}X_2}] (p|_{X_2}) \\
      	&= \top_{\mathcal{S}X_2} (p|_{X_2}) \\
      	&= 0 \\
      	&= \expect{} \circ \mathcal{S}[\bot_{X_1}, \top_{X_2}](p)
      \end{align*}
      Finally, consider the case when $\supp(p) \cap X_1 = \emptyset$ and $\supp(p) \cap X_2 \neq \emptyset$. We have that 
      \begin{align*}
      	[\bot_{\mathcal{S} X_1},\top_{\mathcal{S} X_2}] \circ g_{X_1,X_2}(p) &= [\bot_{\mathcal{S}X_1}, \top_{\mathcal{S}X_2}] (p|_{X_2}) \\
      	&= \top_{\mathcal{S} X_2} (p|_{X_2}) \\
      	&= 0 \\
      	&= \left( \sum_{y = [\bot_{X_1}, \top_{X_2}](x)} p(x) \right) \\
      	&= \expect{} \circ \mathcal{S} [\bot_{X_1}, \top_{X_2}](p) \qedhere
      \end{align*}
    \end{itemize}
  \end{description}  
\end{proof}


It is left to show that the EM-law $\zeta\colon FT\Rightarrow TF$ lifts
to
$\zeta\colon \overline{F}\,\overline{T}\Rightarrow
\overline{T}\,\overline{F}$, where $\overline{F} = \kant{\Lambda^F}$,
$\overline{T} = \kant{\Lambda^T}$ where $\Lambda^T = \{\ev_T\}$,
and where the evaluation maps for $F$ are obtained as described earlier. We namely
prove that its components are all non-expansive, i.e.,
$\Vgraph$-morphisms, commutativity is already known. Using the previous
results we obtain:
\[ \overline{F}\,\overline{T} \stackrel{(1)}{\vle} \overline{TF}
  \stackrel{(2)}{\Rightarrow} \overline{FT} \stackrel{(3)}{=}
  \overline{F}\,\overline{T} \]

\begin{description}
\item[(1)] follows from the inequality in
  \Cref{sec:compositionality}. This implies that the identity
  $\id_{TFX}\colon \overline{T}\,\overline{F}X\to \overline{TF}X$ is
  non-expansive (a $\Vgraph$-morphism), resulting in the natural
  transformation
  $\overline{T}\,\overline{F} \stackrel{\id}{\Rightarrow} \overline{TF}$.
\item[(2)] is implied by the results of this section
  (\Cref{lem:eval-G-comp-zeta-equal-eval-F},
  \Cref{prop:lifting-em-law})
\item[(3)] is guaranteed by the fact that
  $\overline{FT} = \overline{F}\,\overline{T}$ if $F$ is polynomial
  (and we have suitable evaluation maps), hence compositionality holds
  (\Cref{prop:comp-polynomial})
\end{description}

\begin{example}
  We continue \Cref{ex:running-2} and first observe that the
  distributive law given there is obtained by the inductive
  construction given above\todo{B: more details?}.

  For this concrete example fix the label set as a singleton:
  $A = \{a\}$. Consider the coalgebra with states $X = \{x,x',y\}$
  drawn below on the left. The payoff is written next to each state.
  
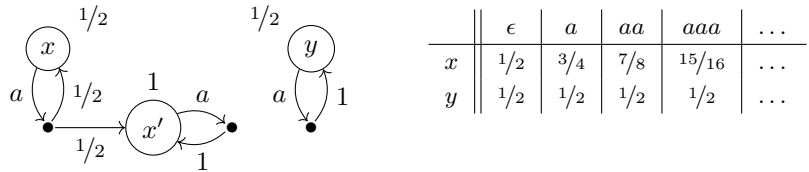
\begin{center}
    \begin{tikzpicture}[scale=0.7]
      \node (Y) at (7,0) [circle,draw,label=160:$\nicefrac{1}{2}$]{$y$};
      \node (Ya) at (7,-1.5) [inner sep=0]{$\bullet$}; 			
      \node (X) at (2,0) [circle,draw,label=20:$\nicefrac{1}{2}$]{$x$};
      \node (Xa) at (2,-1.5) [inner sep=0]{$\bullet$}; 			
      \node (Xp) at (4,-1.5) [circle,draw,label=90:$1$]{$x'$};
      \node (Xpa) at (5.5,-1.5) [inner sep=0]{$\bullet$}; 			
      \draw  [->] (Y) to [bend right] node [left]{$a$} (Ya);
      \draw  [->] (Ya) to [bend right] node [right]{$1$} (Y);	
      \draw  [->] (X) to [bend right] node [left]{$a$} (Xa);
      \draw  [->] (Xa) to [bend right] node [right]{$\nicefrac{1}{2}$} (X);
      \draw  [->] (Xa) to node [below]{$\nicefrac{1}{2}$} (Xp);	
      \draw  [->] (Xp) to [bend left] node [above]{$a$} (Xpa);
      \draw  [->] (Xpa) to [bend left] node [below]{$1$} (Xp);	
    \end{tikzpicture}
    \qquad
    \raisebox{1cm}{
    \begin{tabular}[b]{c||c|c|c|c|c}
      & $\epsilon$ & $a$ & $aa$ & $aaa$ & \dots \\ \hline 
      $x$ & $\nicefrac{1}{2}$ & $\nicefrac{3}{4}$ & $\nicefrac{7}{8}$
      & $\nicefrac{15}{16}$ & \dots \\
      $y$ & $\nicefrac{1}{2}$ & $\nicefrac{1}{2}$ & $\nicefrac{1}{2}$
      & $\nicefrac{1}{2}$ & \dots 
    \end{tabular}}
  \end{center}
  Here, the trace map has the values for states $x,y$ given in the
  table above on the right, leading to the behavioural distance
  $\nu\mathrm{beh}(\delta_x,\delta_y) = \nicefrac{1}{2}$ (the supremum
  of all differences).

  Determinizing the probabilistic automaton above leads to an
  automaton with infinite state space, even if we only consider the
  reachable states (which are probability distributions):
  
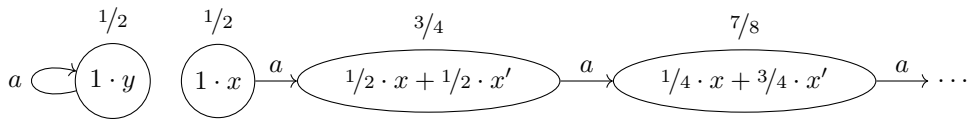
\begin{center}
    \begin{tikzpicture}[scale=0.7]
      \node (X) at (0,0)
      [circle,draw,label=90:$\nicefrac{1}{2}$]{$1\cdot y$};
      \draw  [->] (X) to [loop left] node {$a$} (X);
      \node (X) at (2,0) [circle,draw,label=90:$\nicefrac{1}{2}$]{$1\cdot x$};
      \node (X1) at (6,0) [ellipse,draw,label=90:$\nicefrac{3}{4}$]{$\nicefrac{1}{2}
        \cdot x + \nicefrac{1}{2}\cdot x'$};
      \node (X2) at (12,0) [ellipse,draw,label=90:$\nicefrac{7}{8}$]{$\nicefrac{1}{4} \cdot x + \nicefrac{3}{4}\cdot x'$};
      \node (X3) at (16,0) {$\dots$};
      \draw  [->] (X) to node [above]{$a$} (X1);
      \draw  [->] (X1) to node [above]{$a$} (X2);
      \draw  [->] (X2) to node [above]{$a$} (X3);
    \end{tikzpicture}
  \end{center}
  Our aim is now to define a witness distance of type
  $d\colon TX\times TX\to [0,1]$ that is a post-fixpoint up-to in the
  quantale order and a pre-fixpoint for the order on the reals
  ($d\ge \mathrm{beh}(u(d))$). From this we can infer that
  $\nu \mathrm{beh} \le d$, obtaining an upper bound. We set $d(1\cdot x,1\cdot y) = \nicefrac{1}{2}$, $d(1\cdot x',1\cdot
    y) = \nicefrac{1}{2}$ and $d(p,q) = 1$ for all other pairs of
  probability distributions $p,q$. We now show that $d$ is a
  pre-fixpoint of $\mathrm{beh}$ in the order on the reals, i.e., our
  aim is to prove for all $p,q$ that $d(p,q) \ge \mathrm{beh}(u(d))(p,q)$.
  This is obvious for the cases where $d(p,q) = 1$, hence only two
  cases remain:
  \begin{itemize}
  \item If $p = 1\cdot x$, $q = 1\cdot y$, we have:
    \begin{align*}
      \mathrm{beh}(u(d))(1\cdot x,1\cdot y) &=
      \max\{u(d)(\nicefrac{1}{2}\cdot x+\nicefrac{1}{2}\cdot
      x',y),|\nicefrac{1}{2}-\nicefrac{1}{2}|\} \\
      &= u(d)(\nicefrac{1}{2}\cdot x+\nicefrac{1}{2}\cdot x',y)
      \\
      &\le \nicefrac{1}{2}\cdot d(1\cdot x,1\cdot y) +
      \nicefrac{1}{2}\cdot d(1\cdot x',1\cdot y) \\
      &= \nicefrac{1}{2}\cdot \nicefrac{1}{2} +
      \nicefrac{1}{2}\cdot \nicefrac{1}{2} = \nicefrac{1}{2} =
      d(1\cdot x,1\cdot y)
    \end{align*}
  \item The case $p = 1\cdot x'$, $q = 1\cdot y$ can be shown analogously.
  \end{itemize}
  We are using the following inequalities: (i) $u(d)(p,q) \le d(p,q)$
  (follows from the definition of~$u$); (ii)
  $u(d)(r_1\cdot p_1+r_2\cdot p_2,r_1\cdot q_1+r_2\cdot q_2) \le
  r_1\cdot d(p_1,q_1) + r_2\cdot d(p_2,q_2)$ (metric congruence,
  \full{cf.~\Cref{cor:metric-congruence} in the appendix}\short{see
    \cite{dgkknrw:behavioural-metrics-compositionality-arxiv} for more
    details.}).  This concludes the argument.
\end{example}

Note that here up-to techniques in fact allow us to consider finitary
witnesses for bounds for behavioural distances, even when
the determinized coalgebra has an infinite state space.

\section{Case Study: Transition Systems with Exceptions}
\label{sec:case-study}

In addition to the running example treated in the paper we now
consider a second case study on transition systems with exceptions
that helps to concretely show upper bounds (lower bounds in the
quantalic order) for behavioural distances via appropriate witnesses.

\todo{Wasserstein giving less fine-grained distances explanation. B:
  done, earlier in the running example}
\todo{Up-to functions/convex contextual closure, reference to
  Chatzikokolakis, Palamidessi, Vignudelli and others. B: done
  $\leadsto$ conclusion}

\begin{toappendix}
  We will now show a property of the up-to function $u$ that is
  closely related to the fact that under certain conditions the
  Kantorovich lifting is always bounded by the Wasserstein lifting, as
  shown in \cite{bbkk:coalgebraic-behavioral-metrics}. The requirement
  is part of the notion of a well-behaved evaluation function
  introduced there.
  
  \begin{proposition}
    \label{prop:well-behaved-upper-bound}
    Let $(Y,a,c)$ be an $F$-$T$-bialgebra, $\ev_T$ an evaluation map
    for $T$ and let $u$ be the up-to function (as in
    \Cref{sec:intro-up-to}). Furthermore assume that
    $d\colon Y\times Y\to \VV$ is an object of $\Vgraph$ and
    $t\in T(Y\times Y)$. Then:
    \[ u(d)(a(T\pi_1(t),a(T\pi_2(t))) \vge \ev_T\circ Td(t) \]
    whenever
    \begin{equation}
      d_\VV\circ (\ev_T\times \ev_T)\circ \langle
      T\pi'_1,T\pi'_2\rangle \vge \ev_T\circ Td_\VV.
      \label{eq:well-behaved-w2}
    \end{equation}
    where $\pi_i\colon Y\times Y\to Y$ and
    $\pi'_i\colon \VV\times\VV\to\VV$.
  \end{proposition}

  \begin{proof}
    \begin{eqnarray*}
      u(d)(a(T\pi_1(t),a(T\pi_2(t))) & = &
      \bigsqcup_{a(t_i)=a(T\pi_i(t))} \overline{T}(d)(t_1,t_2) \\
      & \vge & \overline{T}(d)(T\pi_1(t),T\pi_1(t)) \\
      & = & \bigsqcap_{f\in \gamma_Y(d)} d_\VV(\ev_T\circ Tf\circ
      T\pi_1(t),\ev_T\circ Tf\circ T\pi_2(t)) \\
      & \vge & \ev_T\circ Td(t)
    \end{eqnarray*}
    The last inequality follows directly from
    \cite[Lemma~5.25]{bbkk:coalgebraic-behavioral-metrics} whenever
    Condition~\eqref{eq:well-behaved-w2} holds. Note that in
    \cite{bbkk:coalgebraic-behavioral-metrics} the quantale is $\unitQ$
    and $d_\VV$ (called $d_e$ there) is assumed to be symmetric, but
    neither assumption is used in the proof.
  \end{proof}

  \begin{corollary}
    \label{cor:metric-congruence}
    In the setting of \Cref{prop:well-behaved-upper-bound}
    assume that $Y = TX$, $a=\mu_X$ and $\VV=\unitQ$. Furthermore
    consider the monad $T=\mathcal{P}$ ($T = \mathcal{D}$). Then the
    evaluation map $\ev_T = \sup$ ($\ev_T = \expect{}$) satisfies
    Condition~\eqref{eq:well-behaved-w2} and hence we obtain:
    \begin{itemize}
    \item $T=\mathcal{P}$: \[ u(d)(\bigcup_{i\in I}
      X_i,\bigcup_{i\in I} Y_i) \le \sup_{i\in I} d(X_i,Y_i) \]
      whenever $X_i,Y_i\subseteq X$, $i\in I$.
    \item $T = \mathcal{D}$:
      \[ u(d)(\sum_{i\in I} r_i\cdot p_i,\sum_{i\in I} r_i\cdot q_i)
        \le \sum_{i\in I} r_i\cdot d(p_i,q_i) \] whenever
      $p_i,q_i\in \mathcal{D}(X)$, $i\in I$.
    \end{itemize}
  \end{corollary}

  \begin{proof}
    We first prove Condition~\eqref{eq:well-behaved-w2} in both cases
    ($T=\mathcal{P}$, $T=\mathcal{D}$):
    \begin{itemize}
    \item $T=\mathcal{P}$: Let $t\in T(\VV\times \VV)$, i.e.,
      $t\subseteq [0,1] \times [0,1]$ respectively
      $t = \{(s_i,r_i)\mid i\in I, r_i,s_i\in [0,1]\}$. Then
      \begin{eqnarray*}
        d_\VV\circ (\ev_T\times \ev_T)\circ \langle
        T\pi'_1,T\pi'_2\rangle(t) & = &
        d_\VV(\ev_T(T\pi'_1(t)),\ev_T(T\pi'_2(t))) \\
        & = & \sup_{i\in I}  r_i \ominus \sup_{i\in I} s_i \\
        & \le & \sup_{i\in I} (r_i\ominus s_i) \\
        & = & \ev_T\circ Td_\VV(t).
      \end{eqnarray*}
    \item $T=\mathcal{D}$: Let $t\in T(\VV\times \VV)$, i.e.,
      $t\colon [0,1] \times [0,1]\to [0,1]$ is a probability
      distribution respectively
      $t = \sum_{i\in I} p_i\cdot (s_i,r_i)$, written as a formal sum
      where $p_i,r_i,s_i\in[0,1]$ and $\sum_{i\in I} p_i =
      1$. Then
      \begin{eqnarray*}
        d_\VV\circ (\ev_T\times \ev_T)\circ \langle
        T\pi'_1,T\pi'_2\rangle(t) & = &
        d_\VV(\ev_T(T\pi'_1(t)),\ev_T(T\pi'_2(t))) \\
        & = & \sum_{i\in I}  p_i\cdot r_i \ominus \sum_{i\in I}
        p_i\cdot s_i \\
        & \le & \sum_{i\in I} p_i\cdot (r_i\ominus s_i) \\
        & = & \ev_T\circ Td_\VV(t).
      \end{eqnarray*}
    \end{itemize}
    Finally, we spell out the inequalities:
    \begin{itemize}
    \item $T=\mathcal{P}$: Let $t\in T(TX\times TX)$, i.e.,
      $t\subseteq \mathcal{P}(X) \times \mathcal{P}(X)$ respectively
      $t = \{(X_i,Y_i)\mid i\in I, X_i,Y_i\subseteq X\}$. Then
      \begin{eqnarray*}
        u(d)(\bigcup_{i\in I} X_i,\bigcup_{i\in I} Y_i) & = &
        u(d)(\mu_X(T\pi_1(t)),\mu_X(T\pi_2(t))) \\
        & \le & \ev_T\circ Td(t) \\
        & = & \sup_{i\in I} d(X_i,Y_i)
      \end{eqnarray*}
    \item $T=\mathcal{D}$: Let $t\in T(TX\times TX)$, i.e.,
      $t\colon \mathcal{D}(X) \times \mathcal{D}(X)\to [0,1]$ is a
      probability distribution respectively
      $t = \sum_{i\in I} r_i\cdot (p_i,q_i)$, written as a formal sum
      where $r_i\in[0,1]$ with $\sum_{i\in I} r_i = 1$,
      $p_i,q_i\in\mathcal{D}(X)$. Then
      \begin{eqnarray*}
        u(d)(\sum_{i\in I} r_i\cdot p_i,\sum_{i\in I} r_i\cdot q_i)
        & = &
        u(d)(\mu_X(T\pi_1(t)),\mu_X(T\pi_2(t))) \\
        & \le & \ev_T\circ Td(t) \\
        & = & \sup_{i\in I} r_i\cdot d(p_i,q_i)
      \end{eqnarray*}
    \end{itemize}
  \end{proof}
  \end{toappendix}


\todo{Spread probabilistic automata out into running example. B: done}




We consider a case study involving the coproduct, in particular
the polynomial functor $F = [0,1] + (-)^A$ and the monad $T = \pfun$
with evaluation map $\sup$. In a coalgebra of type $c\colon X \to FTX$,
a state can either perform transitions or terminate with some output
value taken from the interval $[0, 1]$; in applications this value
could for example be considered as the severity of the error
encountered upon terminating a computation. Hence the (directed)
distance of two states $x,y$ can intuitively be interpreted as
measuring how much worse the errors reached from a state $y$ are
compared to the errors from $x$.


Note that a state can also terminate without an exception by
transitioning to the empty set. Due to the asymmetry in the
distributive law for the coproduct, a state $X_0\in\pfun(X)$ in the
determinized automaton $c^\#$ will throw an exception as long as one
of the elements $x\in X_0$ throws an exception
($c[X_0]\cap[0,1]\neq\emptyset$). In this case,
$c^\#(X_0)=\sup (c[X_0]\cap[0,1])$ and $X_0$ performs a transition if all
states in $X_0$ do so in the original coalgebra.

The behavioural distance on $TX$ obtained as the greatest fixpoint (in
the quantale order) of $\mathrm{beh}$ can be characterized as follows:
for a set of states $X_0\subseteq X$ and a word $w\in A^*$ let
$\emph{ec}(X_0,w)$ be the length of the least prefix which causes an
exception when starting in $X_0$ (undefined if there are no
exceptions) and let $E(X_0,w)\subseteq [0,1]$ be the set of exception
values reached by that prefix. We define a distance
$d^E_w\colon \mathcal{P}X\times \mathcal{P}X\to [0,1]$ as
$d^E_w(X_1,X_2) = 0$ if $\emph{ec}(X_2,w)$ is undefined,
$d^E_w(X_1,X_2) = 1$ if $\emph{ec}(X_1,w)$ is undefined and
$\emph{ec}(X_2,w)$ defined. In the case where both are defined we set
$d^E_w(X_1,X_2) = \sup E(X_2,w)\ominus \sup E(X_1,w)$ if
$\emph{ec}(X_1,w)=\emph{ec}(X_2,w)$, $d^E_w(X_1,X_2) = 1$ if
$\emph{ec}(X_1,w)>\emph{ec}(X_2,w)$ and $d^E_w(X_1,X_2) = 0$ if
$\emph{ec}(X_1,w)<\emph{ec}(X_2,w)$. Then it can be shown that for
$X_1,X_2\subseteq X$:
\[ \nu \mathrm{beh}(X_1,X_2) = \sup\nolimits_{w\in A^*} d_w^E(X_1,X_2) \]

\todo{B: move this example to appendix?}As a concrete example we take
the label set $A = \{a, b\}$ and the coalgebra $c$ given below, which
is inspired by a similar example in \cite{bp:checking-nfa-equiv}:
\begin{center}
\begin{tikzpicture}[scale=0.9,x=1.6cm, y=1.3cm]
\foreach \x [count=\i] in {x,y,z}{
\node[draw, circle] (\x0) at (0,-\i) {$\x_0$};
\node[draw,circle] (\x1) at (1,-\i)  {$\x_1$};
\node[draw,circle] (\x2) at (2,-\i)  {$\x_2$};
\node (\x3) at (3,-\i)  {$\dots$};
}

\node[draw,circle,label=0:$\nicefrac{1}{4}$] (x4) at (4,-1)  {$x_n$};
\node[draw,circle,label=0:$\nicefrac{1}{3}$] (y4) at (4,-2)  {$y_n$};
\node[draw,circle,label=0:$\nicefrac{1}{2}$] (z4) at (4,-3)  {$z_n$};

\draw[->] (x0) edge node[above]{$a$} (x1) (x1) edge node[above]{$a, b$} (x2) (x2) edge node[above]{$a, b$} (x3) (x3) edge node[above]{$a, b$} (x4);
\draw[->] (y0) edge node[above]{$b$}(y1) (y1) edge node[above]{$a, b$} (y2) (y2) edge node[above]{$a, b$} (y3) (y3) edge node[above]{$a, b$} (y4);
\draw[->] (z0) edge  node[above]{$a,b$} (z1) (z1) edge node[above]{$a, b$} (z2) (z2) edge node[above]{$a, b$} (z3) (z3) edge node[above]{$a, b$} (z4);

\draw[->] (x0) edge[loop left] node{$a, b$} (x0);
\draw[->] (y0) edge[loop left] node{$a, b$} (y0);
\draw[->] (z0) edge[loop left] node{$a, b$} (z0);
\end{tikzpicture}
\end{center}

Here, the outputs of the final states are $c(x_n) = \nicefrac{1}{4}$,
$c(y_n) = \nicefrac{1}{3}$ and $c(z_n) = \nicefrac{1}{2}$, as
indicated. It holds that
$\nu\mathrm{beh}(\{x_0, y_0\}, \{z_0\}) =
\nicefrac{1}{4}$. Intuitively this is true, since the largest distance
is achieved if we follow a path from $x_0$ where $a$ is the $n$-last
letter,\todo{Rephrase. B: done} yielding exception value
$\nicefrac{1}{4}$, while the same path results in the value
$\nicefrac{1}{2}$ from $z_0$, hence we obtain distance
$\nicefrac{1}{2}\ominus \nicefrac{1}{4} = \nicefrac{1}{4}$.

Note that the determinization of the transition system above is of
exponential size and the same holds for a representation of a
post-fixpoint for witnessing an upper bound for behavioural distance.
We construct a $\VV$-valued relation $d$ of linear size witnessing
that $\nu\mathrm{beh}(\{x_0, y_0\}, \{z_0\}) \le \nicefrac{1}{4}$ via
up-to techniques. 
To this end let $d$ be defined by
\[
d(\{x_0, y_0\}, \{z_0\}) = \nicefrac{1}{4} \qquad 
d(\{x_i\}, \{z_i\}) = \nicefrac{1}{4} \qquad
d(\{y_i\}, \{z_i\}) = \nicefrac{1}{6}
\]
and distance~$1$ for all other arguments.

It suffices to show that $d$ is a pre-fixed point of $\mathrm{beh}$
up-to $u$ (wrt. $\le$). We can use the property that
$u(d)(X_1 \cup X_2, Y_1 \cup Y_2) \le \max\{d(X_1, Y_1), d(X_2,
Y_2)\}$ (\full{cf. \Cref{cor:metric-congruence} in the
  appendix}\short{see
  \cite{dgkknrw:behavioural-metrics-compositionality-arxiv} for more
  details}). Now the claim follows from unfolding the fixpoint equation
  by considering $a$- and $b$-transitions:
\begin{align*}
  \mathrm{beh}(u(d))(\{x_0, y_0\}, \{z_0\}) &= \max\{u(d)(\{x_0, x_1, y_0\}, \{z_0, z_1\}), u(d)(\{x_0, y_0, y_1\}, \{z_0, z_1\})\}\\
  &\le \max\{d(\{x_0, y_0\}, \{z_0\}), d(\{x_1\}, \{z_1\}), d(\{y_1\}, \{z_1\})\}
  = \nicefrac{1}{4}
\end{align*}
The arguments for $d(\{x_i\}, \{z_i\})$, $d(\{y_i\}, \{z_i\})$ are
similar, concluding the proof.

\section{Conclusion}
\label{sec:conclusion}

\subparagraph*{Related work.} While the notion of Kantorovich distance
on probability distributions is much older, Kantorovich liftings in a
categorical framework have first been introduced in
\cite{bbkk:coalgebraic-behavioral-metrics} and have since been
generalized, for instance to codensity liftings
\cite{kwrk:composing-codensity-bisimulations} or to lifting fuzzy lax
extensions \cite{ws:metrics-fuzzy-lax}.

A coalgebraic theory of up-to techniques was presented in
\cite{bppr:general-coinduction-up-to} and has been instantiated to the
setting of coalgebraic behavioural metrics in
\cite{bkp:up-to-behavioural-metrics-fibrations-journal}. The latter
paper concentrated on Wasserstein liftings, which leads to a
significantly different underlying theory. Furthermore, Wasserstein
liftings are somewhat restricted, since they rely only on a single
evaluation map and on couplings (which sometimes do not exist, making
it difficult to define more fine-grained metrics). We are not aware of
a way to handle the two case studies directly in the Wasserstein
approach.

Setting up the fibred adjunction (\Cref{sec:adjunction}) and the
definition of the Kantorovich lifting (\Cref{sec:def-kantorovich})
has some overlap to \cite{bgkmfsw:logics-coalgebra-adjunction} and
the recent \cite{kwrk:composing-codensity-bisimulations}. Our
focus is primarily on showing fibredness (naturality) via a quantalic
version of the extension theorem.

The aim of \cite{kwrk:composing-codensity-bisimulations} is on
combining behavioural conformance via algebraic operations, which is
different than our notion of compositionality via functor
liftings. There is some similarity in the aims of both papers, namely
the lifting of distributive laws and the motivation to study up-to
techniques. From our point of view, the obtained results are largely
orthogonal: while the focus of
\cite{kwrk:composing-codensity-bisimulations} is on providing $n$-ary
operations for composing conformances and games and it provides a more
general high-level account, we focus concretely on compositionality of
functor liftings (studying counterexamples and treating the special
case of polynomial functors), giving concrete recipes for lifting
distributive laws and applying the results to proving upper bounds via
up-to techniques.

Our up-to techniques are a form of up-to convex contextual closure:
here they arise as specific instances of a general construction, but
they have already been investigated in the Wasserstein setting
\cite{bkp:up-to-behavioural-metrics-fibrations-journal} and earlier in
\cite{cgpx:generalized-bisim-metrics}. Similar constructions are
studied in \cite{mpp:quantitative-algebraic-reasoning}.


\subparagraph*{Future work.} Our aim is to better understand the
metric congruence results employed in the case studies, comparing them
with the similar proof rules in
\cite{mpp:quantitative-algebraic-reasoning}. Compositionality of
functor liftings fails in important cases
(cf. \Cref{expl:non-compositional}), which could be fixed by using
different sets of evaluation maps such as the Moss liftings in
\cite{ws:metrics-fuzzy-lax}. We also plan to study case studies
involving the convex powerset functor
\cite{bss:power-convex-algebras}. Finally, we want to develop witness
generation methods, by constructing suitable post-fixpoints up-to
on-the-fly, similar to \cite{bblm:on-the-fly-exact-journal}.

\bibliographystyle{plainurl}
\bibliography{tarball/references-bib2doi}

\end{document}